\documentclass[12pt,twoside]{article}
\usepackage[mathscr]{eucal}
\usepackage{amsmath,amsfonts,amssymb,amsthm}

\usepackage[matrix,arrow]{xy}
\usepackage{times}
\usepackage{graphicx}
\usepackage{epstopdf}
\usepackage{hyperref}

\advance\textheight\topskip
\renewcommand{\tilde}{\widetilde}
\renewcommand{\hat}{\widehat}
\newtheorem{prop}{Proposition}[section]

\newtheorem{definition}[prop]{Definition}
\newtheorem{theorem}[prop]{Theorem}
\newtheorem{corollary}[prop]{Corollary}

\renewcommand{\d}{\partial}

\newcommand{\RR}{\mathbb{R}}

\def\cG{\mathcal{G}}
\def\cH{\mathcal{H}}

\numberwithin{equation}{section} \makeatletter
\@addtoreset{equation}{section}

\begin{document}

\def\mytitle{Hamiltonian surface charges using external sources.}

\pagestyle{myheadings} \markboth{\textsc{\small Troessaert}}{%
  \textsc{\small Hamiltonian surface charges using external sources.}}
  \addtolength{\headsep}{4pt}

\begin{centering}

  \vspace{1cm}

  \textbf{\Large{\mytitle}}



  \vspace{1.5cm}

  {\large C\'edric Troessaert$^a$}

\vspace{.5cm}

\begin{minipage}{.9\textwidth}\small \it \begin{center}
   Centro de Estudios Cient\'ificos (CECs)\\ Arturo Prat 514,
   Valdivia, Chile \\ troessaert@cecs.cl \end{center}
\end{minipage}

\end{centering}

\vspace{1cm}

\begin{center}
  \begin{minipage}{.9\textwidth}
    \textsc{Abstract}. 
    In this work, we interpret part of the boundary conditions as external
    sources in order to solve the integrability problem present in
    the computation of surface charges associated to gauge symmetries in the
    hamiltonian formalism. We start by describing the hamiltonian structure of
    external symmetries preserving the action up to a transformation of the
    external sources of the theory. We then extend these results to the
    computation of surface charges for field theories with non-trivial
    boundary conditions.
    \end{minipage}
\end{center}

\vfill

\noindent
\mbox{}
{\scriptsize$^a$Laurent Houart postdoctoral fellow.}

\thispagestyle{empty}
\newpage

\begin{small}
{\addtolength{\parskip}{-1.5pt}
 \tableofcontents}
\end{small}
\newpage

\section{Introduction}
\label{sec:introduction}

In field theories, associating conserved generators to gauge symmetries is a
long standing problem. The main issue is that the bulk part of the generator
will be proportional to the constraints of the theory: evaluated on solutions
the associated conserved quantity will be zero. At first sight, this does not
seem so bad. However, it means that, for instance, the notion of electric
charge in Maxwell's theory or energy in Einstein's theory both disappear.

In hamiltonian formalism, a partial solution to the problem has been
developped in \cite{regge_role_1974} and \cite{brown_poisson_1986}. The idea is that
generators must be supplemented with a boundary term in order to be
differentiable. The on-shell value of this boundary term is then the associated
conserved quantity. Applying this to Einstein's theory in 4D, one recovers the
ADM value for the mass \cite{regge_role_1974,arnowitt_republication_2008}.

When applying this idea to compute conserved quantities, one has to select a
suitable set of boundary conditions. This choice is the key factor: if the
boundary conditions are too restrictive, this will reduce the set of available
symmetries; if the boundary conditions are too lax, the differentiability
condition on the generators will be too strong and the set of symmetries for
which we can associate conserved quantities will be small. Looking for a good
set of boundary conditions is searching an optimum point of these two
tendencies. This problem is often refered as the integrability problem of
surface charges as the selection of the boundary term to form a differentiable
generator involves solving an integrability equation on the space of field
configurations. The references \cite{henneaux_asymptotically_2004,
henneaux_asymptotic_2007} contain a few examples of integrability
problems in the hamiltonian formalism.

\vspace{5mm}

In this work, we want to present a different approach. The
first part of the idea is to treat boundary conditions as external sources.
This is certainly reasonable and has already been used, for instance in AdS/CFT
(see \cite{witten_anti-sitter_1998} and subsequent literature). The second part is to allow symmetries to act on the
sources. The obtained symmetries are not really symmetries of the theory and
will no give rise to conserved quantities. They are symmetries that send a
solution of the equations of motion for one value of the sources to a solution
with a different value of the sources. They are symmetries between different
theories and, in this work, we will call them external symmetries. We will see
that they are generated by canonical generators and that the algebra of these
generators forms a possibly extended representation of the algebra of the
external symmetries.

Treating boundary conditions as external sources introduces two layers of
boundary conditions. The external layer describes all field configurations for
all possible values of the boundary conditions (Dirichlet, ...). This
is the set of conditions that must be preserved by the external symmetries.
The internal layer describes the field configurations for each theory: for each
specific
value of the boundary conditions. This is the set of boundary conditions
satisfied by the dynamical part of the fields and is related to the
differentiability condition of the generators. In other words, we decoupled
the two effects of the boundary conditions in the integrability problem and
introduced a lot more possibilities.

\vspace{5mm}

The first part of this work contains a definition and study of the notion of external
symmetries in hamiltonian theories with a finite number of degrees of freedom.
We show that, under reasonable assumptions, these symmetries are generated by
canonical generators. We also introduce a modified poisson bracket to compute
the algebra of these associated generators. This algebra forms a possibly
extended representation of the algebra of the external symmetries.

In the second part, we study hamiltonian field theories without gauge freedom.
We promote boundary conditions to external sources and use the results we
obtained earlier to define external symmetries in this case. As before, we
introduce a modified poisson bracket and compute the algebra of the
generators. We then provide an example by applying the results to a scalar
field theory.

The last part contains the study of gauge theories. Due to possible
interactions between dynamical fields and lagrange multipliers through
boundary conditions, we have to keep the lagrange multipliers explicitly in
our analysis. To this end, we start with a small generalisation of standard
results concerning conserved quantities in gauge theories with a finite
number of degrees of freedom. We then combine this with the notion of external
symmetries developed in the previous sections. As before, we compute the
algebra of the generators of external symmetries and show that it forms a
possibly extended representation of the algebra of the associated symmetries.

\section{Hamiltonian theory with external sources}
\label{sec:extsources}
In this section, we will describe the hamiltonian theory in presence
of external sources. We will focus on mechanical systems with a
finite number of degrees of freedom.

A general system has the following action:
\begin{equation}
\label{eq:actsources}
S[q^i, p_i; j^\alpha] = \int dt \, \left\{\dot q^i p_i - H(t,q^i, p_i ;
  j^\alpha, d_t j^\alpha, ... , d_t^k j^\alpha) \right\},
\end{equation}
where $j^\alpha$ is an external source: it is not varied when
deriving the equations of motions. If we only work with 
transformations that preserve exactly the form of the action, sources included, the usual hamiltonian theory of charges is well
behaved. We are able to associate conserved canonical generators and compute their
algebra for each fixed value of $j$. The general picture is a symmetry group
for $j=0$ that, when $j \ne 0$, is broken to the subgroup
preserving this particular value of $j$.

There exists systems where we can keep all the symmetries if we
allow the transformations to act on the sources. For instance
electromagnetism in presence of an external current
\begin{equation}
S[A_\mu; j^\nu] = \int d^nx \left(-\frac{1}{4} F_{\mu\nu} F^{\mu\nu} +
A_\mu j^\mu\right),
\end{equation}
is invariant under the full Poincar\'e algebra if we allow the 
transformations to act on
$j^\mu$. However, these transformations don't really preserve the
action, they send one value of the sources to a different one: they
are symmetries between different problems. Our goal in this section is
to see how we can extend the hamiltonian theory of charges to these
generalized symmetries. 

An external symmetry $\delta_G$ of the action (\ref{eq:actsources}) will be defined as a
transformation of the form 
\begin{eqnarray}
\delta_G q^i & = & Q^i(t,q, d_t q, ... ,p, d_t p , ..., j, d_t j, ...), \\ 
\delta_G p_i &=& P_i(t,q, d_t q, ... ,p, d_t p , ..., j, d_t j, ...), \\ 
\delta_G j^\alpha &=& J^\alpha
(t,q, d_t q, ... ,p, d_t p , ..., j, d_t j, ...),
\end{eqnarray}
that preserves the integrand of the action up to a time derivative or a
function of the sources
\begin{equation}
\label{eq:invaction}
\delta_G  \left\{\dot q^i p_i - H(t,q^i, p_i ;
  j^\alpha, d_t j^\alpha, ... , d_t^k j^\alpha) \right\} =
\frac{d}{dt} U(t,q,p;j,,...) + V_G(j, d_tj, ...).
\end{equation}
The most important property in the usual
case is that a symmetry of the action is a symmetry of the equations
of motion. This is coming from the identity
\begin{equation}
\label{eq:commutvareuler}
\delta_G \frac{\delta}{\delta z^A} L = \frac{\delta}{\delta
  z^A} \delta_G L + \sum ^\infty_{k=0} (-d_t)^k \left(\frac{\d Z^B}{ \d
    (d_t)^k z^A } \frac{\delta L}{\delta z^B} \right),
\end{equation}
where $z^A=(q^i, p_i)$ and $Z^A=(Q^i, P_i)$.
If $\delta_G L = d_t U $, this identity reduces to
\begin{equation}
\delta_G \frac{\delta}{\delta z^A} L = \sum ^\infty_{k=0} (-d_t)^k
\left(\frac{\d Z^B}{ \d
    (d_t)^k z^A} \frac{\delta L}{\delta z^B} \right),
\end{equation}
which is zero on the EOM $\frac{\delta L}{\delta z^A}\approx 0$. In our
case, the identity \eqref{eq:commutvareuler} becomes
\begin{equation}
\delta \frac{\delta}{\delta z^A} L = \sum_{k=0}^\infty (-d_t)^k \left(\frac{\d Z^B}{ \d
    (d_t)^k z^A } \frac{\delta L}{\delta z^B} \right) + \sum ^\infty_{k=0} (-d_t)^k \left(\frac{\d J^\alpha}{ \d
    (d_t)^k z^A } \frac{\delta L}{\delta j^\alpha} \right).
\end{equation}
This is zero on the
equations of motion if $J^\alpha$ is independent of the dynamical
fields. The sources have to transform without involving the dynamical
fields. 

This is a natural restriction to impose on the
transformations. In this case, one problem (the variational problem
for one value of $j^\alpha$) will be send by the transformation to another
problem (same variational problem for a different value of $j^\alpha$). If
the transformation of $j^\alpha$ is allowed to depend on the dynamical fields,
all the possible values of the dynamical fields for one value of $j^\alpha$
will be send to different values of the sources, to different
problems. This will completely destroy our variational
theory. In the following, we will only work with
transformations where $J^\alpha$ is independent of $z^A$. Another useful way
to encode this information is:
\begin{equation}
	\label{eq:preservj}
	\left[ \delta, \delta_G\right] j^\alpha = 0 \quad \forall \delta
	\quad \text{s.t.} \quad \delta j^\alpha = 0.
\end{equation}

For any action, we have trivial symmetry transformations
\begin{equation}
\delta_M z^A = M^{AB} (\frac{\delta L}{\delta z^B}) \qquad \delta_M
j^\alpha = 0,
\end{equation}
where the operator $M^{AB}$ is an anti-self-adjoint operator:
\begin{equation}
X_A M^{AB}\left(Y_B \right) = -Y_A M^{AB}\left(X_B \right) +
\frac{d}{dt} \Xi_M(X,Y).
\end{equation}
They can be used to remove all dependence of $Z^A$ in the time derivatives of the
fields $z^A$. Because of this, the most general transformation we will
consider is of the form
\begin{equation}
\label{eq:transfo}
\delta_G z^A = Z^A (t,z^B, j^\alpha, d_t j^\alpha , ... , (d_t)^l
j^\alpha), \quad \delta_G j^\alpha = J^\alpha (t, j^\beta, d_t j^\beta , ... , (d_t)^l j^\beta).
\end{equation}
Equation (\ref{eq:invaction}), can be rewritten as
\begin{equation}
\dot Q^i p_i + \dot q^i P_i - \frac{\d H}{\d q^i} Q^i - \frac{\d H}{\d
  p_i} P_i - \delta_G^j H =  \frac{\d U}{\d q^i} \dot q^i +\frac{\d
  U}{\d p_i} \dot p_i +\frac{\d U}{\d t}  + \delta^j_t U + V_G,
\end{equation}
where we have introduced the following notations
\begin{eqnarray}
\delta^j_G & = & J^\alpha \frac{\partial}{\partial j^\alpha} +d_t J^\alpha
\frac{\partial}{\partial d_t j^\alpha}+d^2_t J^\alpha
\frac{\partial}{\partial d^2_t j^\alpha}+...\\
\delta^j_t & = & d_tj^\alpha \frac{\partial}{\partial j^\alpha} +d^2_t j^\alpha
\frac{\partial}{\partial d_t j^\alpha}+d^3_t j^\alpha
\frac{\partial}{\partial d^2_t j^\alpha}+...
\end{eqnarray}
This equation is valid for all values of $q^i$, $p_i$, $\dot q^i$ and
$\dot p_i$. It can be decomposed into the following equations:
\begin{eqnarray}
\label{eq:GgenQ}
\frac{\d Q^j}{\d q^i}p_j + P_i & =&  \frac{\d U}{\d q^i},\\
\label{eq:GgenP}
\frac{\d Q^j}{\d p_i} p_j & = & \frac{\d  U}{\d p_i}, \\
\d_t Q^i p_i + \delta^j_t Q^i p_i - \frac{\d H}{\d q^i} Q^i - \frac{\d H}{\d
  p_i} P_i - \delta_G^j H &=&  \frac{\d U}{\d t}  + \delta^j_t U + V_G.
\end{eqnarray}
The first two equations imply that the transformation of the hamiltonian variables $(q^i,
p_i)$ is symplectic and has as generator $G=Q^j  p_j-U$. Using this on the last equation leads to 
\begin{equation}
\label{eq:hamilconserv}
\partial_t G 
+ \{ G,H\} + \delta^j_t G - \delta^j_G H  =  V_G,
\end{equation}
where we have introduced the usual poisson bracket given
by:
\begin{equation}
\{ F,G\} = \frac{\d F}{\d q^i} \frac{\d G}{\d p_i} - \frac{\d G}{\d q^i} \frac{\d F}{\d p_i} = \frac{\d F}{\d z^A} \sigma^{AB}\frac{\d G}{\d z^B},
\end{equation}
with $\sigma^{AB}$ the poisson structure.

\vspace{5mm}

We have the following result: 
\begin{theorem}A transformation of the form
\begin{equation}
\delta_G z^A= Z^A (t,z,j,\partial_t j,... , \partial^l_t j) \qquad \text{and} \qquad
\delta_G j^\alpha = J^\alpha  (t,j,\partial_t j,... , \partial^l_t j),
\end{equation}
preserves the action in the sense \eqref{eq:invaction} if and only if there
exists a generator $G$ such
that it is the hamiltonian generator of the transformation of the canonical
variables and it satisfies \eqref{eq:GgenQ}, \eqref{eq:GgenP} and 
(\ref{eq:hamilconserv}).
\end{theorem}
\noindent Such transformations will be called external symmetries. Equation \eqref{eq:hamilconserv} is the equivalent of the conservation of $G$ in the
absence of external sources. It can be rewritten as:
\begin{equation}
	\delta_t G[z; j] = \delta^j_G H[z; j] + V_G[j], \quad \delta_t \equiv
	\frac{\d}{\d t} +\sigma^{AB} \frac{\d H}{\d z^B} \frac{\d}{\d z^A} +
	\delta^j_t.
\end{equation}
Giving the generator $G$ is not enough to describe the transformation,
one must also supply the variation of the sources $J^\alpha$ or
equivalently the operator $\delta^j_G$. In the following, when we are
referring to a transformation generated by $G$, we will assume that both
$G$ and $\delta^j_G$ are known. The generator $G$ is
defined up to a function of the sources: both $G$ and $G + K[j]$
generate the same transformation. 

\vspace{5mm}

Let's compute the generator associated to the commutator of two
transformations $(F, \delta^j_F)$ and $(G, \delta^j_G)$:
\begin{eqnarray}
\delta_{[F,G]} z^A&=& \delta_F \left\{ z^A, G[z; j]\right\} - F \leftrightarrow G\nonumber\\
&=& \left\{ \left\{  z^A,F[z; j]\right\}, G[z; j]\right\} + \left\{ z^A,
  \left\{G[z; j],F[z; j] \right\}+\delta^j_F  G[z; j]\right\} - F \leftrightarrow G\nonumber\\
&=&\left\{ z^A,
  \left\{G[z; j],F[z; j] \right\}+\delta^j_F  G[z; j]-\delta^j_G  F[z; j]\right\}.\label{eq:commuvar}
\end{eqnarray}
The canonical generator of $\delta_{[F,G]}z^A$ is given by:
\begin{equation}
	\left\{G[z; j],F[z; j] \right\}+\delta^j_F  G[z; j]-\delta^j_G  F[z; j].
\end{equation}
Following this, we introduce a new
poisson Bracket for the generators in the presence of external
sources:
\begin{equation}
\left\{F,G \right\}_j\equiv\left\{F,G \right\}+\delta^j_G  F-\delta^j_F  G.
\end{equation}
As we said earlier, the generators are associated with specific
transformations of the sources, the bracket of the generator only is
not well-defined. It should be extended to couples $(G, \delta^j_G)$. The couple
generating the transformation $\delta_{[F,G]}$ is given by:
\begin{equation}
\left[(F, \delta_F^j), (G,\delta_G^j) \right] = \left(\left\{G, F\right\}_j, [\delta^j_F, \delta^j_G]\right).
\end{equation}
This bracket is obviously antisymmetric and it can be checked that it
satisfies the Jacobi identity. If we identify the hamiltonian $H$ with
the generator of time translations, we must supply it with the right
variation of the source namely $\delta_H j^\alpha\equiv d_t j^\alpha$, which
implies $\delta^j_H = \delta^j_t$. Using our new poisson Bracket, equation
(\ref{eq:hamilconserv}) becomes
\begin{equation}
\label{eq:hamilconserv2}
\partial_t G + \{ G,H\}_j =  V_G[j].
\end{equation}

In the usual case, when considering symmetries of the action
preserving the sources $\delta^j_G=0$, we also have the possibility to add this extra
term $V_G$. In order to preserve the EOM, we need $V_G$ to be a
constant. The analysis done above goes through and we obtain this constant on
the right hand side of the conservation equation
\eqref{eq:hamilconserv2} with the unmodified poisson bracket. In this
case, the constant can be absorbed in the generator by adding to it a term
of the form $-t V_G$ allowing us to recover the usual result. For
the more general case considered in this work, due to the presence 
of the sources and their
arbitrary time dependence, we cannot get rid of this extra term.

\vspace{5mm}

The set of external symmetries (\ref{eq:transfo}) form an
algebra $\cG$ of 
symmetries of the set of theories parametrized by $j^\alpha$. To any
of these transformations, we can associate a couple $(G,\delta^j_G)$. These couples associated to the
commutator
\begin{equation}
\label{eq:commut}
\left[(F, \delta_F^j), (G,\delta_G^j) \right] = \left(\{G,F\}_j, [ \delta_F^j, \delta_G^j] \right)
\end{equation}
form a representation of $\mathcal{G}$ (The only
property missing is that $\left(\{G,F\}_j, [ \delta_F^j, \delta_G^j]\right)$ satisfies
(\ref{eq:hamilconserv2}) which is proven in appendix
\ref{sec:appsources}). This representation has
room for extensions. This is coming from the fact that we can
add any function of the sources to a generator $G$. 
\begin{theorem}
If $G_i$ forms a
generating set of the algebra $\mathcal{G}$, we have in general
\begin{equation}
\left(\{G_2,G_1\}_j, [ \delta_1^j, \delta_2^j] \right) = \left(
  G_{[1,2]} + K_{1,2}[j], \delta^j_{[1,2]} \right),
\end{equation}
with $K_{1,2}$ antisymmetric and
\begin{equation}
K_{[1,2],3} + \delta^j_3 K_{1,2} + cyclic = 0.
\end{equation}
\end{theorem}
\noindent The cyclic identity of the extension $K_{1,2}$ comes from the Jacobi identity of
the modified bracket. This is an abelian extension based on the representation of $\mathcal{G}$ on
the sources $j$ \cite{de_azcarraga_lie_1998}.

\vspace{5mm}

As an application, let's consider a point particle on which we act
with an external force $f_i(t)$:
\begin{equation}
S[q^i, p_i; f_i] = \int dt \left( \dot q^i p_i -\frac{p_ip^i}{2m} 
+ q^i f_i \right),
\end{equation}
where indices are raised and lowered with the Kronecker delta.
The equations of motion are
\begin{eqnarray}
\dot q^i &=& \frac{p^i}{m},\\
\dot p_i &=& f_i,
\end{eqnarray}
and, if we remove $p$, we obtain the famous Newton equations
\begin{equation}
m\,  d_t^2 q^i = f^i.
\end{equation}
The presence of this source term in the action spoils galilean
symmetries but using the above results, we can still associate
generators to them. The galilean transformations are:
\begin{eqnarray}
\delta_{a,v,\omega} q^i &= & a^i + tv^i - \omega^i_{\phantom i j} q^j,\\
\delta_{a,v,\omega} p_i &= & mv_i + \omega^j_{\phantom i i} p_j,\\
\delta_{a,v,\omega} f_i &= & \omega^j_{\phantom i i} f_j,
\end{eqnarray}
where $a^i$, $v^i$ and $\omega_{ij}$ with $\omega_{(ij)}=0$
respectively parametrize translations, boosts and rotations.
These transformations are of the form (\ref{eq:transfo}) and
preserve the action in the sense \eqref{eq:invaction}:
\begin{eqnarray}
\delta_{a,v,\omega} L =\frac{d}{dt}(m q^iv_i)+ (a^i+tv^i) f_i.
\end{eqnarray}
We can then associate non-conserved generators
\begin{eqnarray}
G_{a,v,\omega} &=& \left( a^i - \omega^i_{\phantom i j} q^j \right)
p_i + tv^ip_i - m v_iq^i,\\
d_t G_{a,v,\omega} &=& \left( a^i +tv^i- \omega^i_{\phantom i j} q^j \right)
f_i.
\end{eqnarray}
Part of the evolution of $G_{a,v,\omega}$ is just a rewriting of the well-known
results concerning the evolution of the momentum and angular
momentum of a point particle in presence of an external force. In
this case, the new poisson bracket of the generators is the same than
the old one and the algebra closes with a central
extension between the boosts and the translations~\cite{de_azcarraga_lie_1998}:
\begin{gather}
	\left\{ G_1, G_2\right\}_j = (\hat a^i - \hat \omega^i_{\phantom i j} q^j)
	p_i + t\hat v^ip_i - m \hat v_iq^i + m (a^i_1v_{2i} - a^i_2v_{1i}),
\end{gather}
where
\begin{gather}
	\hat a^i = \omega^i_{2j}a^j_1 - \omega^i_{1j}a^j_2, \quad \hat v^i =
	\omega^i_{2j}v^j_1 - \omega^i_{1j}v^j_2,\\
	\hat \omega^i_{\phantom ij} = \omega^i_{2k}\omega^k_{1j} -
	\omega^i_{1k}\omega^k_{2j}.
\end{gather}


\section{Boundary conditions as external sources}
\label{sec:bound-cond-as}

In this section, we will use the previous results to generalize the
hamiltonian theory of charges to symmetries not preserving
boundary conditions. The first subsection introduces the theoretical objects needed and
gives the results concerning external symmetries when boundary conditions are
treated as external sources. The second subsection contains an application: we
compute the external symmetries for a scalar field contained in a sphere with
Dirichlet boundary conditions.

From here on, we will use the conventions introduced in appendix A of
\cite{barnich_surface_2008} in order to describe spatial field configurations. 

\subsection{Generalized differentiable functionals}
\label{sec:gendifffunc}

We will consider hamiltonian systems defined on a
manifold $\Sigma$ with boundaries $\d\Sigma$ at a finite distance. Functionals
of the
canonical fields $z^A(x)$ are assumed to have the usual poisson bracket
\begin{equation}
\label{eq:PBnormal}
\left\{ F, G \right\} = \int_\Sigma d^nx \, \frac{\delta F}{\delta
  z^A}\sigma^{AB}\frac{\delta G}{\delta z^B},
\end{equation}
with $\sigma^{AB}$ a constant, anti-symmetric and invertible matrix. We will also
assume a set of boundary conditions on $z^A$. The results we will
describe in this section can be
extended easily to boundaries at infinity and asymptotic conditions.

\vspace{5mm}

The Hamiltonian theory of charges for field theories introduced in
\cite{regge_role_1974} and further developed in \cite{brown_poisson_1986} is
based on the idea that only differentiable generators are allowed in
the poisson bracket. A differentiable generator $G[z]$ is a functional
such that:
\begin{itemize}
\item the boundary term produced in its variation is zero:
\begin{equation}
\delta G = \int_\Sigma d^nx \, \frac{\delta G}{\delta z^A} \delta z^A,
\end{equation}
where $\delta z^A$ respects the boundary conditions imposed on $z^A$.
\item its associated hamiltonian vector field $\delta_G$ with
\begin{equation}
	\delta_G = \d_{(i)}G^A \frac{\d^S}{\d z^A_{(i)}}, \qquad G^A =
	\sigma^{AB}\frac{\delta G}{\delta z^B},
\end{equation}
preserves the boundary conditions on $z^A$.
\end{itemize}
With this definition, the hamiltonian theory is well behaved. However,
there are cases where the above conditions are too restrictive. In the
following, we will relax the second condition by treating part of the
boundary conditions on $z^A$ as external sources.

\vspace{5mm}

We will consider boundary conditions on the bulk fields $z^A$
parametrized by a set of non dynamical boundary fields $\zeta^\alpha$:
\begin{equation}
\label{eq:BC-gen}
\left.\chi_I^\alpha (z)\right\vert_{\d \Sigma} = \zeta^\alpha.
\end{equation}
We will assume the operators $\chi_I^\alpha$ to be local and 
independent of time.
The standard cases are Dirichlet or Newmann boundary conditions. 
In order to guarantee smooth solutions, we have to supplement these "internal"
boundary conditions with extra "external" boundary conditions:
\begin{gather}
	\label{eq:BCext}
	\chi_E^{k,\alpha} = (d_t - \delta_t)^k \chi_I^\alpha \quad
	\forall k>0, \qquad
	\delta_t = \frac{\d}{\d t} + \delta_H,\\
	\chi_E^{k,\alpha} \vert_{\d\Sigma} = 0 \quad \forall k>0 \quad
	\Leftrightarrow \quad (d_t^k -
	\delta_t^l)\chi_I^\alpha\vert_{\d\Sigma}= 0 \quad \forall k>0,
\end{gather}
where $d_t$ is the total time derivative and $\delta_H$ is the variation
generated by the hamiltonian.
Once a value of the boundary field $\zeta^\alpha$ has been chosen, these extra
boundary conditions become:
\begin{equation}
	\delta^k_t \chi^\alpha_I\vert_{\d\Sigma} = d_t^k \zeta^\alpha \quad
	\forall k>0,
\end{equation}
which usually are boundary conditions on normal
derivatives of $z^A$. However, if $\zeta^\alpha$ is allowed to vary, they
don't impose any restriction on configurations $z^A(k)$ at fixed time. When it
doesn't lead to confusion, we will use $\chi_E$ to denote the set of external
boundary conditions $\chi_E^{k,\alpha}$.

In the following, we will call variations of the fields $\delta z^A$
preserving the internal boundary conditions the variations
satisfying:
\begin{equation}
	\label{eq:preservint}
	\delta \chi_I^\alpha\vert_{\d \Sigma} = 0,\quad \delta
	\chi_E\vert_{\d \Sigma} = 0,\qquad\forall
	z \text{ s.t. } \chi_E\vert_{\d\Sigma} = 0.
\end{equation}
In general, if we restrict ourselves to configurations associated to a fixed
value of $\zeta^\alpha$, the set
of variations obtained is bigger than the one defined in \eqref{eq:preservint}. In the following, we will only consider
operators $\chi_I^\alpha$ such that these two sets are equal.

Inspired by the results of the previous section, we extend the definition of
differentiable generator to:
\begin{definition}
A generalized differentiable generator $G[z]$ is a functional such that,
for all variation $\delta z^A$ preserving the internal boundary conditions, we have 
\begin{equation}
\label{eq:diffgenG}
\delta G = \int_\Sigma d^nx \frac{\delta G}{\delta z^A} \delta z^A,\qquad \forall
z \text{ s.t. } \chi_E\vert_{\d\Sigma} = 0,
\end{equation}
and
\begin{equation}
\label{eq:preservsourceG}
\delta_G \chi_E\vert_{\d\Sigma} = 0, \quad \left[\delta, \delta_G\right]
\chi_I^\alpha(z) \vert_{\d\Sigma} =0,
	\qquad \forall
	z \text{ s.t. } \chi_E\vert_{\d\Sigma} = 0,
\end{equation}
where $\delta_G z^A \equiv G^A=\sigma^{AB} 
\frac{\delta G}{\delta z^B}$ is the transformation generated by $G$.
\end{definition}
\noindent The first condition \eqref{eq:diffgenG} guarantees that the variational principle
generating the transformation $\delta_G z^A$ is well-defined for each
value of the boundary conditions $\zeta^\alpha$:
\begin{equation}
	S_G[z^A; \zeta^\alpha] = \int dt\int_\Sigma d^nx \, 
		\left( z^A \sigma_{AB}\d_s z^B - G[z]\right),
\end{equation}
where $\sigma_{AB}$ is the inverse of $\sigma^{AB}$. The second
condition \eqref{eq:preservsourceG} is the requirement that the
transformation of the sources depends on the sources only: it is the
equivalent of equation \eqref{eq:preservj} of the previous section. 
The differentiable generators in the sense of Regge-Teitelboim
\cite{regge_role_1974} are 
the generalized differentiable generators preserving the internal boundary conditions:
$\delta_G\chi_I^\alpha(z)\vert_{\d\Sigma}=0$.

\vspace{5mm}

By analogy with the results of the previous section, we want to define the
following modified bracket of generalized differentiable functionals
$F$ and $G$:
\begin{equation}
\left\{F, G \right\}_\zeta = \left\{ F, G\right\} + \delta_G^\zeta F -
\delta_F^\zeta G,
\end{equation}
where $\delta^\zeta$ is the variation of the functionals only hitting
the sources $\zeta^\alpha$. However, whereas in the previous section
the separation between dynamical variables and sources was easily
done, in this case, $\zeta^\alpha$ and $z^A$ are linked by 
boundary conditions which means that the action of $\delta^\zeta_G$ is 
hard to identify.

To this end, let's consider a general variation $\delta z^A$ preserving only
the external boundary conditions $\delta \chi_E\vert_{\d\Sigma}=0$. Equation \eqref{eq:diffgenG} implies that the boundary
term only contains variations of the source $\zeta^\alpha$:
\begin{equation}
\label{eq:genbound1}
\delta G = \int_\Sigma d^nx \frac{\delta G}{\delta z^A} \delta z^A +
\oint_{\d \Sigma}(d^{n-1}x)_i \Theta^i_G [\delta \zeta^\alpha]. 
\end{equation}
Let's now consider variations of the form
\begin{equation}
\label{eq:genbound2}
\delta_\epsilon z^A \equiv \eta_\epsilon(x) \, \delta z^A, 
\end{equation}
where $\eta_\epsilon$ are smooth functions that are zero in a neighbourhood of the boundary and
such that
\begin{equation}
\lim_{\epsilon \rightarrow 0} \eta_\epsilon(x) = 1, \qquad \forall x
\in \Sigma \backslash \d \Sigma.
\end{equation}
For all values of $\epsilon\ne 0$, $\delta_\epsilon z^A$ preserves the
internal boundary conditions \eqref{eq:preservint} which implies 
\begin{equation}
\lim_{\epsilon \rightarrow 0}\left( \delta - \delta_\epsilon\right) G =
  \oint_{\d \Sigma} (d^{n-1}x)_i \Theta^i_G [\delta \zeta^\alpha].
\end{equation}
It means that the source part of the variation of $G[z]$ is only
encoded in the boundary term. With this, we define the modified
bracket of generalized differentiable functionals
$F$ and $G$ as:
\begin{equation}
\label{eq:PBmodif}
\left\{F, G \right\}_\zeta = \left\{ F, G\right\} + \oint_{\d \Sigma}
(d^{n-1}x)_i \Theta^i_F [\delta_G \zeta^\alpha] - \oint_{\d \Sigma}
(d^{n-1}x)_i \Theta^i_G [\delta_F  \zeta^\alpha].
\end{equation}
The boundary terms $\Theta_F$ and $\Theta_G$ are defined by equation
\eqref{eq:diffgenG} for both $F$ and
$G$. Because the modification only
concerns boundary terms, the variation generated by $\left\{ F,
  G\right\}_\zeta$ is given by
\begin{equation}
	\label{eq:commutvar}
\delta_{\{F,G\}_\zeta} = \left[\delta_G, \delta_F \right].
\end{equation} 
We also see that, for differentiable functionals in the sense of Regge-Teitelboim
\cite{regge_role_1974}, both variations $\delta_G\zeta$ and
$\delta_F\zeta$ are zero and the bracket \eqref{eq:PBmodif} reduces to
the unmodified poisson bracket \eqref{eq:PBnormal}. 

\begin{theorem}
The modified bracket \eqref{eq:PBmodif} is a well-defined poisson
bracket.
\end{theorem}
\begin{proof}
	It is manifestly anti-symmetric and it satisfies the Jacobi
identity. An easy way to see that is to recognize that this bracket is
a particular case of the one introduced in \cite{bering_putting_2000} where we reduced
the set of available functionals. For completeness, we also give a direct
proof in appendix \ref{sec:jacobiMB}. The only thing still left to prove is that the modified
bracket of two generalized differential functionals is a
generalized differentiable functional. The second property
\eqref{eq:preservsourceG} is easy to check using the Jacobi identity
for variations:
\begin{equation}
\left[ \delta_{\{F,G\}_\zeta}, \delta\right] =
\left[\delta_G,[\delta_F, \delta] \right] - \left[ \delta_F,[\delta_G, \delta] \right].
\end{equation}
For the first condition \eqref{eq:diffgenG}, the analysis is done in appendix
\ref{sec:diffbrack}.
\end{proof}

\vspace{5mm}

Now that we have defined our additional structures, we will see
how they describe the external symmetries of the system:
\begin{equation}
S[z] = \int dt \left(\int_\Sigma d^nx \, \frac{1}{2}z^A \sigma_{AB}\dot z^B - H[z]\right).
\end{equation}
We will assume that a set of boundary conditions $\chi_I^\alpha, \chi_E$ has been
selected such that the action is well-defined under variations preserving the
internal boundary conditions. In that case, one can check that
the Hamiltonian $H[z]$ is a generalized differentiable functional with 
$\delta_H\zeta^\alpha = \dot \zeta^\alpha$. 

In analogy with the previous
section, we define:
\begin{definition}
An external symmetry $\delta_G z^A= Z^A(t,z)$ of the action is a
transformation such that
\begin{gather}
\label{eq:gentrans}
\delta_G \chi_E \vert_{\d\Sigma} = 0,\quad  [\delta_G, \delta]\chi_I^\alpha(z)\vert_{\d \Sigma} = 0,
\end{gather}
and
\begin{gather}
\label{eq:gensymm1}
\delta_G S[z] = \int dt\left\{ \frac{d}{dt} \int_\Sigma d^nx\, U(t,z) +
  \oint_{\d\Sigma}(d^{n-1}x)_i V_G^i[z]\right\},\\
\label{eq:gensymm2}
  \delta \oint_{\d\Sigma}(d^{n-1}x)_i V_G^i(z) = 0,
\end{gather}
for all $\delta$ preserving the internal boundary conditions. 
\end{definition}
\noindent The last condition \eqref{eq:gensymm2} means that the boundary term $V_G$ only depends on the non-dynamical boundary
fields $\zeta^\alpha$. As in the previous section, this is a symmetry
between different systems corresponding to different values of the
boundary conditions $\zeta^\alpha$. Remark that compared to the previous
section, we don't have to specify the variation of the sources as it is
obtained by continuity from the bulk variation.

\begin{theorem}
	\label{theo:fieldsymm}
A transformation of the form \eqref{eq:gentrans},
\eqref{eq:gensymm1} is an external symmetry of the action
\eqref{eq:gensymm2} if and only if there exists a generalized
differentiable functional $G[z]$ such that:
\begin{gather}
\delta_G z^A = \left\{ z^A, G\right\} = \sigma^{AB} \frac{\delta
  G}{\delta z^B},\\
\label{eq:geneconserv}
\d_t G + \left\{ G, H\right\}_\zeta = \oint (d^{n-1}x)_i \, V_G^i(z),\qquad
  \delta \oint_{\d\Sigma}(d^{n-1}x)_i V_G^i(z) = 0,
\end{gather}
where $\d_t$ only hits the explicit dependence in $t$ and $\delta$ span all
variations preserving
the internal boundary conditions.
\end{theorem}
\begin{proof}
	Let's assume that $\delta_G$ is an external symmetry of the action, developing the
	LHS of \eqref{eq:gensymm1}, we get
	\begin{multline}
		\label{eq:gensymmpros}
		\int dt \left\{ \int_\Sigma d^nx \left( \frac{1}{2}\sigma_{AB} Z^A
				\dot z^B+  \frac{1}{2}\sigma_{AB} z^A d_t Z^B -
			Z^A \frac{\delta H}{\delta z^A}\right) - \oint_{\d
		\Sigma} (d^{n-1}x)_i \Theta^i_H[\delta_G z] \right\}\\ = 
		\int dt \left\{ \frac{d}{dt} \int_\Sigma d^nx\, U(t,z) +
  \oint_{\d\Sigma}(d^{n-1}x)_i V_G^i(z)\right\}.
	\end{multline}
	Introducing the following functional:
	\begin{equation}
		G[z] = \int_\Sigma d^nx \, \left\{\frac{1}{2} \sigma_{AB} z^A
		Z^B - U(t, z) \right\},
	\end{equation}
	 equation \eqref{eq:gensymmpros} can be written as
	\begin{multline}
		\int dt	\frac{d G}{d t} = \int dt \left\{ \int_\Sigma d^nx
		\left( -\sigma_{AB} Z^A \dot z^B +
			Z^A \frac{\delta H}{\delta z^A}\right)\right.\\ \left.
			+ \oint_{\d
			\Sigma} (d^{n-1}x)_i \left(\Theta^i_H[\delta_G z] +
		V_G^i(z)\right)\right\},
	\end{multline}
	where we can expand the LHS to
	\begin{equation}
		\int dt	\frac{d G}{d t} = \int dt \left\{ \frac{\d G}{\d t} 
			+\int_\Sigma d^nx \frac{\delta G}{\delta z^A} 
			\dot z^A + \oint_{\d \Sigma} (d^{n-1}x)_i
		\Theta^i_G[\dot z] \right\}.
	\end{equation}
	In the bulk, $\dot z^A$ is arbitrary, this implies
	\begin{equation}
		\frac{\delta G}{\delta z^B} = -\sigma_{BA} Z^A \qquad
		\Rightarrow \qquad Z^A = \sigma^{AB}\frac{\delta G}{\delta
		z^B} ,
	\end{equation}
	which means that $G$ is the canonical generator of the transformation.
	Putting everything together, the equality \eqref{eq:gensymmpros}
	becomes
	\begin{equation}
		\label{eq:gensymmprosII}
		\int dt \left[ \frac{\d G}{\d t} + \left\{G, H \right\}_\zeta
		- \oint_{\d \Sigma} (d^{n-1}x)_iV^i(z)\right] = \int dt
		\oint_{\d\Sigma} (d^{n-1})_i \Theta^i_G[\delta_H z - \dot z].
	\end{equation}
	By construction, the transformation $\delta_H z - \dot z = (\delta_t - d_t)z = \hat \delta
	z$ preserves all boundary
	conditions. Moreover, apart from these preservation conditions, it is
	completely arbitrary: it is an arbitrary variation preserving the
	internal boundary conditions. As the LHS of
	\eqref{eq:gensymmprosII} is independant of $\dot z^A$, we have
	\begin{equation}
		\oint_{\d\Sigma} (d^{n-1})_i \Theta^i_G[\delta z] = 0,
	\end{equation}
	for all variations $\delta z^A$ preserving the internal boundary conditions.
	The generator $G$ is a generalized differentiable generator. The LHS
	of \eqref{eq:gensymmprosII} being zero for all intervals of
	integration $[t_0, t_1]$ leads to:
	\begin{equation}
		\d_t G + \left\{ G, H\right\}_\zeta = \oint (d^{n-1}x)_i \,
		V_G^i(z),
	\end{equation}
	which is what we wanted.

	For the other direction, let's assume that we have a generalized
	differentiable generator $G$ satisfying \eqref{eq:geneconserv}.  The
	variation of the action under the transformation generated by $G$,
	$\delta_G z^A = G^A$, is 
	\begin{eqnarray}
		\delta_G S & = & \int dt \int d^nx \left\{
			\frac{1}{2}\sigma_{AB}G^A\dot z^B +
		\frac{1}{2}\sigma_{AB}z^A d_t G^B - \delta_G H\right\}\\
		& = & \int dt \left\{\frac{d}{d t}\left( \int_\Sigma d^nx\,
		\frac{1}{2}\sigma_{AB} z^A G^B \right) + \int_\Sigma d^nx \,
		\sigma_{AB} G^A\dot z^B\right.\nonumber \\
		&& \left.\qquad + \left\{ G, H\right\}_\zeta -
		\oint_{\d\Sigma} (d^{n-1})_i \Theta^i_G[\delta_H z] \right\}\\
		& = & \int dt \left\{\frac{d}{d t}\left( \int_\Sigma d^nx
		\frac{1}{2}\sigma_{AB} z^A G^B \right) - \int_\Sigma d^nx\, 
		\dot z^A\frac{\delta G}{\delta z^A}\right.\nonumber \\
		&& \left.\qquad -  \d_t G  +
		\oint_{\d\Sigma} (d^{n-1})_i\left(V^i_G(z)-
	\Theta^i_G[\dot z]\right) \right\}\\
		& = & \int dt \left\{\frac{d}{d t}\left( \int_\Sigma d^nx
		\frac{1}{2}\sigma_{AB} z^A G^B - G \right)+
		\oint_{\d\Sigma} (d^{n-1})_iV^i_G(z)\right\}
	\end{eqnarray}
	which means that $\delta_G$ is an external symmetry of the action. Between the
	second and the third line we used equation \eqref{eq:geneconserv} and
	the fact that $\hat \delta z^A =\delta_Hz^A - \dot z^A$ preserves the
	internal boundary
	conditions.
\end{proof}

\vspace{5mm}

\begin{theorem}
The external symmetries of the action form an algebra. The
associated generalized differentiable generators with the modified
bracket form an extended representation of this algebra:
\begin{equation}
\left\{ G_2, G_1\right\}_\zeta = G_{[1,2]} + \oint_{\d\Sigma}
(d^{n-1}x)_i K^i_{1,2}(z),
\end{equation}
where the extra boundary term is invariant under all variations preserving the
internal boundary conditions and satisfies
\begin{gather}
	K^i_{1,2} = - K^i_{2,1},\\
\oint_{\d\Sigma}
(d^{n-1}x)_i \left(K^i_{[1,2],3}+ \delta_3 K^i_{1,2} + cyclic\right) = 0.
\end{gather}
\end{theorem}
\noindent As before, this extension is in general non-central as it is based on the
representation of the algebra of external symmetries on the
boundary fields $\zeta$.
\begin{proof}
	The modified poisson bracket of two generalized differentiable
	generators is a generalized differentiable generator. The only thing
	we need is to check the modified conservation law
	\eqref{eq:geneconserv}. Let's assume $F$ and $G$ generate external
	symmetries of
	the action, we get
	\begin{eqnarray}
		\d_t \left\{F, G\right\}_\zeta + \left\{\left\{ F, G
		\right\}_\zeta , H\right\}_\zeta &=& \left\{ \d_tF +
		\left\{F, H\right\}_\zeta, G\right\}_\zeta + \left\{ F, \d_tG +
		\left\{G, H\right\}_\zeta\right\}_\zeta\nonumber\\
		&=& \left\{ \oint_{\d\Sigma} (d^{n-1}x)_i V_F^i, G\right\}_\zeta 
		+ \left\{F, \oint_{\d\Sigma}(d^{n-1}x)_i V_G^i
		\right\}_\zeta\nonumber \\
		&=& \oint_{\d\Sigma} (d^{n-1})_i \left(\delta_G V_F^i -
		\delta_F V_G^i \right).
	\end{eqnarray}
	The boundary term in the last line is zero under an arbitrary
	variation preserving the internal boundary conditions.

	We proved that the external symmetries of the action form an algebra. Let's
	consider $G_1$, $G_2$ and $G_{[1,2]}$ respectively the generators of
	the external symmetries of the action $\delta_1$, $\delta_2$ and
	$\delta_{[1,2]}$. As we showed, $G_{[1,2]}$ and $\{G_2, G_1\}$
	generate the same transformation. Their difference is a generalized
	differentiable generator producing a zero variation: it is a boundary
	term invariant under all variations preserving the internal boundary
	conditions. We have
	\begin{gather}
		\left\{G_2, G_1\right\}_\zeta = G_{[1,2]} + \oint_{\d\Sigma}
		(d^{n-1}x)_i K^i_{1,2}(z),\\
		\delta \oint_{\d\Sigma} (d^{n-1}x)_i K^i_{1,2}(z) = 0,
	\end{gather}
	for all variations $\delta z^A$ preserving the internal boundary conditions.
	The antisymmetry of the poisson bracket  means that
	\begin{equation}
		K^i_{1,2} = - K^i_{2,1}.
	\end{equation}
	If we have a third external symmetry $\delta_3$ then the Jacobi identity 
	for the modified poisson bracket gives the cyclic
	identity:
	\begin{equation}
		\oint_{\d\Sigma} (d^{n-1})_i \left[ K^i_{1,[2,3]} + \delta_3
		K^i_{1,2} + cyclic\right] = 0.
	\end{equation}
\end{proof}

\subsection{Example: scalar field}
\label{sec:exampl-scal-field}

In this section, we will use the theory of a single scalar field as an example for the results we presented in the previous section. The
hamiltonian action is given by
\begin{equation}
S[\phi, \pi] = \int dt\int_{\Sigma} d^3x \left\{ \pi \dot \phi -
  \frac{1}{2}\left( \frac{1}{\sqrt{g}}\pi^2 + \sqrt{g}g^{ij}\d_i\phi
    \d_j \phi\right)\right\},
\end{equation}
where $g_{ij}$ is the metric on $\Sigma$. We will consider $\Sigma$ a
ball in $\RR^n$ with its boundary being a $n-1$ sphere: $\d \Sigma =
S_{n-1}$. We will take the metric to be flat and use
spherical coordinates $x^i = r, x^A$:
\begin{equation}
g_{ij}dx^idx^j = dr^2 + r^2\gamma_{AB} dx^Adx^B.
\end{equation}
The covariant derivative associated to $g_{ij}$ will be denoted $D_i$,
in particular, we will have $D_i\pi = \d_i \pi - \Gamma^j_{ij}\pi$ as
$\pi$ is a density. We will use Dirichlet boundary conditions:
\begin{equation}
\label{eq:bcscalarfieldI}
\chi_I(\phi, \pi) = \phi, \qquad \chi_I(\phi, \pi)\vert_{\d\Sigma}= \bar \phi(t, x^A).
\end{equation}
To have smooth solutions, we also need to impose the external boundary
conditions \eqref{eq:BCext}. When the boundary field $\bar\phi$ is fixed, they become:
\begin{gather}
\label{eq:bcscalarfieldII}
\left.\frac{1}{\sqrt g}\pi \right\vert_{\d \Sigma}= \frac{d}{dt} \bar \phi, \quad \left.(D^iD_i)^k \phi \right\vert_{\d \Sigma} =
\frac{d^{2k}}{dt^{2k}} \bar \phi, \quad \left.(D^iD_i)^k \left(\frac{1}{\sqrt g}\pi \right) \right\vert_{\d \Sigma}= \frac{d^{2k+1}}{dt^{2k+1}} \bar \phi,
\end{gather}
for all integers $k>0$. The laplacian in $n$ dimensions $D^iD_i$ can be
decomposed into its radial and angular part. Doing this, the
external boundary conditions become boundary conditions on some
combinations of radial derivatives of the dynamical fields.

\vspace{5mm}

A scalar field in $\RR^{n+1}$ has the full Poincar\'e
symmetry. If we restrict the theory to $\Sigma$, this symmetry is
broken to the subgroup preserving the boundary conditions $\bar\phi(t,
x^A)$. For general values of $\bar \phi$, the resulting symmetry group
is trivial. 

Let's now consider the external symmetries we defined in the previous section
and see what subalgebra of
Poincar\'e is preserved. Poincar\'e transformations take the form
\begin{equation}
\delta_\xi \phi = \frac{\xi^\perp}{\sqrt g}\pi + \xi^i \d_i \phi, \qquad \delta_\xi \pi
= \sqrt g g^{ij}D_i \left(\xi^\perp\d_j \phi \right) + \d_i \left(
  \xi^i \pi\right),
\end{equation}
where $\xi^i-tD^i\xi^\perp$ is a time independent killing vector of
$g_{ij}$ and $\xi^\perp$ satisfies $\d_t \xi^\perp=0$ and
$D_iD_j\xi^\perp=0$. One can easily check that those combine into
$\xi^\mu=(\xi^\perp, \xi^i)$ to form a
killing vector of Minkowski $ds^2 =- dt^2
+g_{ij}dx^idx^j$. As $\delta_\xi$ is a symmetry of the action, we have
$[d_t - \delta_t, \delta_\xi]=0$. This leads to:
\begin{eqnarray}
	\delta_\xi \chi_I &=& \frac{\xi^\perp}{\sqrt g}d_t\bar\phi + \xi^A \d_A
	\bar\phi +\xi^r \d_r \phi+ \xi^A \d_A(\chi_I - \bar\phi) +
	\frac{\xi^\perp}{\sqrt g} (\delta_t - d_t)\chi_I,\\
	\delta_\xi \chi_E^k &=& (d_t - \delta_t)^k \delta_\xi \chi_I,\nonumber\\
	& = &\xi^r (d_t - \delta_t)^k \d_r \phi+ \xi^A \d_A
	\left[(d_t - \delta_t)^k \chi_I\right] -
	\frac{\xi^\perp}{\sqrt g} (d_t - \delta_t)^{k+1}\chi_I.
\end{eqnarray}
The transformation $\delta_\xi$ satisfy the boundary conditions of an external
symmetry \eqref{eq:preservsourceG} if and only if $\xi^r\vert_{\d\Sigma}=0$.
From the Poincar\'e transformations, only linear combinations of 
time translation and 
rotations satisfy those two conditions. This seems natural as they
are the only transformations preserving  $\Sigma$. The associated
generalized differential generator is
\begin{gather}
G[\xi^\perp, \xi^i] = \int_\Sigma d^nx \, \left(\xi^\perp \cH_\perp + \xi^A \cH_A\right),\\
\cH_\perp = \frac{1}{2}\left( \frac{1}{\sqrt{g}}\pi^2 + \sqrt{g}g^{ij}\d_i\phi
    \d_j \phi\right),\qquad \cH_i = \pi \d_i \phi,
\end{gather}
with $\xi^\perp$ a constant and $\xi^A$ a killing vector of the
$n-1$ sphere. A variation of the fields $\delta$ preserving only the external boundary
conditions $\delta\chi_E\vert_{\d\Sigma} = 0$ leads to
\begin{gather}
\label{eq:genrestrictPoincare}
\delta G[\xi^\perp, \xi^i] = \int_\Sigma d^nx \, \left( \frac{\delta
    G}{\delta \phi} \delta \phi + \frac{\delta
    G}{\delta \pi} \delta \pi\right) + \oint_{\d\Sigma} (d^{n-1}x)_r
\, \sqrt g \xi^\perp D^r \phi  \, \delta  \bar \phi,\\
\Rightarrow \quad \Theta^r_\xi[\delta \bar \phi] = \sqrt g \xi^\perp D^r \phi  \,\delta \bar \phi.
\end{gather}

\vspace{5mm}

The poisson bracket of two generators of the form
\eqref{eq:genrestrictPoincare} is given by
\begin{multline}
\left\{G[\xi^\perp, \xi^i], G[\eta^\perp, \eta^i]\right\} =
\int_{\Sigma} d^nx \, \left\{ (\xi^B \d_B \eta^A - \eta^B \d_B \xi^A)
  \cH_A \right\}\\
+ \oint_{\d \Sigma} (d^{n-1})_r \sqrt g \left(\eta^\perp \xi^A -
  \xi^\perp \eta^A \right) D^r\phi \, \d_A \phi.
\end{multline}
We see that the poisson bracket produces an extra boundary term. However, as
expected, this boundary term is killed if we use the modified bracket
\eqref{eq:PBmodif}:
\begin{equation}
\left\{G[\xi^\perp, \xi^i], G[\eta^\perp, \eta^i]\right\}_\zeta =
\int_{\Sigma} d^nx \, \left\{ (\xi^B \d_B \eta^A - \eta^B \d_B \xi^A)
  \cH_A \right\}.
\end{equation}
The algebra closes without extension. The hamiltonian being given by $H =
G[1,0]$ leads to
\begin{equation}
\d_t G[\xi^\perp, \xi^i] + \left\{G[\xi^\perp, \xi^i], H\right\}_\zeta = 0,
\end{equation}
which proves that these transformations are external symmetries of
the theory.

In this section, we have restricted our analysis to the Poincar\'e 
transformations for clarity. On top of considering the boundary 
conditions as sources, we also could have treated the metric as a
source. The external symmetries would then include all diffeomorphisms
preserving the form of the boundary.

\section{Surface charges for gauge theories}
\label{surface_charges}

Before studying the external symmetries of gauge field theories, we will spend
some time studying symmetries and conserved charges for gauge theories with a
finite number of degrees of freedom.

\subsection{Symmetries of gauge theories}
\label{sec:symgaugetheo}

The theories we will work with are of the form:
\begin{equation}
	S[z^A, \lambda^a] = \int dt\, \left(\frac{1}{2} z^A \sigma \dot z^B-
	h(z) - \lambda^a \phi_a\right),
\end{equation}
where $\phi_a$ are first-class constraints and $h$ is a first-class function. The
symmetries and associated conserved quantities of this class of theories are
studied in exercise 3.24 of \cite{henneaux_quantization_1992}. We will now review some of
the results obtained in this reference and introduce some new concepts that will
be needed in section \ref{sec:GFT}.

As in \cite{henneaux_quantization_1992}, let's consider a transformation of the form
\begin{equation}
	\label{eq:symgauge}
	\delta_G z^A = Z^A\left(t, z, \lambda, \dot\lambda, ...,
	\overset{(k)}{\lambda}\right), \quad \delta_G \lambda^a = \Lambda^a\left(t, z, \lambda, \dot\lambda, ...,
		\overset{(k)}{\lambda}\right).
\end{equation}
It is a symmetry of the action if and only if there exists a generator
$G\left(t, z, \lambda, \dot\lambda, ..., \overset{(k)}{\lambda}\right)$ such
that
\begin{gather}
	Z^A=\sigma^{AB} \frac{\d}{\d z^B} G,\\
	\label{eq:symgaugecons}
	\frac{D}{Dt}G + \left\{G, H\right\} = \Lambda^a\phi_a,
\end{gather}
where
\begin{equation}
	\label{eq:gaugebigdt}
	\frac{D}{Dt} = \frac{\d}{\d t} + \sum_{l=0} \overset{(l+1)}{\lambda^a}
	\frac{\d}{\d \overset{(l)}{\lambda^a}} \quad \text{and} \quad
	H=h+\lambda^a\phi_a.
\end{equation}
From \eqref{eq:symgaugecons}, one can show that the dependence in $\lambda$
of the generator $G$ is proportional to the constraints:
\begin{equation}
	\label{eq:symmdeplambda}
	G\left(t, z, \lambda, \dot\lambda, ..., \overset{(k)}{\lambda}\right)
	= \bar G(t, z) + g^a\left(t, z, \lambda, \dot\lambda, ...,
	\overset{(k)}{\lambda}\right) \phi_a.
\end{equation}
The conservation of $\bar G$ then takes the usual form:
\begin{equation}
	\frac{\d}{\d t} \bar G + \left\{\bar G, H\right\} \approx 0,
\end{equation}
where $\approx$ denotes the equality on the constraint surface. Because of
this, when studying symmetries of gauge theories, we usually restrict ourselves
to cases for which both $Z^A$ and $G$ are independent of $\lambda$ as the
potential dependence in the lagrange multipliers can always be absorbed by a
gauge transformation.

When studying gauge field theories in the next section, this restriction might lead to
problems. In some cases, the boundary conditions can create a link between the
two kind of fields: $z^A$ and $\lambda^a$. In order to deal with this it will
be easier to allow an explicit dependence on $\lambda^a$ in the variation of the
canonical variables $\delta_G z^A$. Because of this, we will spend the rest of
this section studying the algebra of the generators associated to symmetries
of the form \eqref{eq:symgauge}.

\vspace{5mm}

The first observation is that these symmetries don't form a closed subalgebra.
The problem is that the commutator of two symmetries given by
\begin{eqnarray}
	\left[\delta_1, \delta_2\right] z^A & = & Z_1\frac{\d}{\d z^B} Z^A_2 +
	\sum_{l=0} d_t^l\Lambda^b_1 \frac{\d}{\d
	\overset{(l)}{\lambda^b}}Z_2^A - (1\leftrightarrow 2),\\
	\left[\delta_1, \delta_2\right] \lambda^a & = & Z_1\frac{\d}{\d z^B} \Lambda^a_2 +
	\sum_{l=0} d_t^l\Lambda^b_1 \frac{\d}{\d
	\overset{(l)}{\lambda^b}}\Lambda^a_2 - (1\leftrightarrow 2),
\end{eqnarray}
contains explicit dependence in the time derivative of $z^A$ through the terms $d_t^l\Lambda^a$.
The way out is that we are studying equivalence classes of symmetries
where two symmetries are equivalent if their difference is a trivial symmetry:
\begin{equation}
	\delta_1 \sim \delta_2 \Leftrightarrow \left\{ \begin{array}{rcl}
			\delta_1 z^A - \delta_2 z^A& = &
			M^{AB}\left(\frac{\delta L}{\delta z^B}\right) - 
			M^{\dagger bA}\left(\frac{\delta L}{\delta
			\lambda^b}\right),\\
			\delta_1 \lambda^a - \delta_2 \lambda^a& = &
			M^{aB}\left(\frac{\delta L}{\delta z^B}\right) + 
			M^{ab}\left(\frac{\delta L}{\delta
			\lambda^b}\right),\\
		\end{array}\right.
\end{equation}
with
\begin{gather}
	M^{..}(F) = \sum_{l=0}^{k} M^{..}_l d_t^l F, \quad
	M^{\dagger ..}(F) = \sum_{l=0}^{k} \left(-d_t\right)^l\left(
	M^{..}_l F\right),\\
	M^{\dagger AB} = -M^{BA}, \qquad M^{\dagger ab} = - M^{ba}.
\end{gather}
Two equivalent symmetries will lead to conserved quantities that are equal on
the equations of motion. One can show that: $[\delta_1, \delta_2] \sim
\delta_{[1,2]}$ where
\begin{eqnarray}
	\delta_{[1,2]} Z^A & = & Z_1\frac{\d}{\d z^B} Z^A_2 +
	\sum_{l=0} \delta_t^l\Lambda^b_1 \frac{\d}{\d
	\overset{(l)}{\lambda^b}}Z_2^A\\
	&& \quad + \sum_{l=1}\sum_{i=0}^{l-1}\sigma^{AB}\frac{\d}{\d z^B}
	\left(\delta^i_t \Lambda^b_1\right)\, (-\delta_t)^{l-i-1}\left(\phi_a \frac{\d}{\d
	\overset{(l)}{\lambda^b}}\Lambda^a_2\right) - (1\leftrightarrow 2),
	\nonumber \\
	\delta_{[1,2]} \lambda^a & = & Z_1\frac{\d}{\d z^B} \Lambda^a_2 +
	\sum_{l=0} \delta_t^l\Lambda^b_1 \frac{\d}{\d
	\overset{(l)}{\lambda^b}}\Lambda^a_2 - (1\leftrightarrow 2),\\
	\delta_t F & = & \frac{D}{Dt}F + \{F, H\}.
\end{eqnarray}
The transformation $\delta_{[1,2]}$ is a symmetry of the action of the form
\eqref{eq:symgauge} with the following generator:
\begin{equation}
	G_{[1,2]}= \left\{G_2, G_1\right\} + \sum_{l=0} \left( \delta^l_t
		\Lambda_1^a \frac{\d}{\d \overset{(l)}{\lambda^a}}G_2 - \delta^l_t
	\Lambda_2^a \frac{\d}{\d \overset{(l)}{\lambda^a}}G_1\right).
\end{equation}

At $t$ fixed, one can treat $\overset{(i)}{\lambda^a}$ as
independent variables. If we define:
\begin{equation}
	\tilde \delta_G z^A = Z^A_G, \quad \tilde \delta_G
	\overset{(i)}{\lambda^a} = \delta^i_t \Lambda_G^a, \quad \forall i,
\end{equation}
the above results can be rewritten:
\begin{gather}
	\delta_{[1,2]} \lambda^a = \tilde\delta_1 \Lambda_2^a - \tilde
	\delta_2 \Lambda_1^a,\\
	G_{[1,2]} = \left\{ G_2, G_1\right\}^g,
\end{gather}
where
\begin{eqnarray}
	\left\{G_1, G_2\right\}^g & \equiv & \{G_1, G_2\} +
	\tilde\delta_2^\lambda G_1 - \tilde\delta_1^\lambda G_2 \\
	& = & \tilde \delta_2 G_1- \tilde\delta_2 G_1 - \{G_1, G_2\}.
\end{eqnarray}
The notation $\tilde\delta^\lambda$ denotes the part of the variation only hitting
the dependence in $\lambda^a$ and its time derivatives. Defining
$\delta_H
\lambda^a = \dot \lambda^a$, the conservation condition
\eqref{eq:symgaugecons} becomes
\begin{equation}
	\frac{\d}{\d t} G + \{G, H\}^g = 0.
\end{equation}

These results are very similar to what we obtained when studying external
sources in section \ref{sec:extsources}. However, there are a few differences. The
first one is that the transformation of the lagrange multipliers in
\eqref{eq:symgauge} can depend on the canonical variables. The second one is
that the bracket induced on the couples $(G, \delta_G^\lambda)$ does not
satisfy the Jacobi identity. If we define:
\begin{gather}
	\left[(G_1, \delta^\lambda_1), (G_2, \delta^\lambda_2)\right] \equiv
	\left(\{G_2, G_1\}^g, \delta^\lambda_{[1,2]}\right),\\
	\tilde\delta_{[1,2]}\overset{(i)}{\lambda^a} = \delta^i_t\left(\tilde \delta_1 \Lambda^a_2 - \tilde\delta_2
	\Lambda^a_2\right),
\end{gather}
then
\begin{equation}
	\left[\left[(G_1,\delta^\lambda_1), (G_2, \delta_2^\lambda)\right],
	(G_3,\delta_3^\lambda)\right] + cycl = (G_J, \delta^\lambda_J),
\end{equation}
where
\begin{equation}
	G_J = \sum_{l=0} \left([\tilde \delta_2, \delta_t^i] \Lambda_1^a -
	[\tilde \delta_1, \delta_t^i] \Lambda_2^a\right) \frac{\d}{\d
	\overset{(l)}{\lambda^a}} G_3 + cycl.
\end{equation}
If $(G_3, \delta_3^\lambda) = (H, \delta_H^\lambda)$ then this expression
simplifies to $G_J=0$. This is another way of saying that, if both $(G_1,
\delta^\lambda_1)$ and $(G_2, \delta^\lambda_2)$ generate symmetries of the
action, then $[(G_1,\delta^\lambda_1),(G_2, \delta^\lambda_2)]$ also generates
a symmetry of the action. If all three couples $(G_1, \delta^\lambda_1)$,
$(G_2, \delta^\lambda_2)$ and $(G_3, \delta^\lambda_3)$ generate symmetries of
the action, we have $G_J \approx 0$ as expected.
In the case where the generators $G_1$ and $G_2$ are independent of the
lagrange multipliers, the associated transformations $\delta^\lambda_1$ and
$\delta^\lambda_2$ don't matter and we have
\begin{equation}
	\{G_1, G_2\}^g = \{G_1, G_2\}.
\end{equation}

\subsection{Gauge field theories}
\label{sec:GFT}

This section is dedicated to the generalisation of the results of section \ref{sec:bound-cond-as}
to gauge field theories. We will see how our previous analysis can solve some
integrability problems in the definition of surface charges associated to
gauge transformations. The main idea
will be to combine the results obtained in the previous subsection with those
obtained in section \ref{sec:gendifffunc}.

We will consider gauge theories of the form:
\begin{gather}
\label{eq:gaugeaction}
S[z,\lambda] = \int dt \left\{\int_\Sigma d^nx  \frac{1}{2} z^A \sigma_{AB}
\dot z^B - H[z, \lambda]\right\},\\
	H[z,\lambda] = \int_\Sigma d^nx \left(h(z) +
	\lambda^a\phi_a\right) +
		\oint_{\d\Sigma}(d^{n-1}x)_i b^i_H(z,\lambda),
\end{gather}
where $\phi_a$ are first-class constraints and $h$ is a first-class hamiltonian.
In general, both the dynamical fields $z^A$ and the lagrange multipliers
$\lambda^a$ will
have non-trivial internal boundary conditions:
\begin{equation}
\label{eq:gaugebcI}
\chi_I^\alpha (z, \lambda) \vert_{\d \Sigma} = \zeta^\alpha.
\end{equation}
On top of these, we will impose the equations of motion and all their
derivatives on the boundary:
\begin{equation}
	\label{eq:gaugeBCEOM}
	\left.\d_{(i)} \phi_a \right\vert_{\d\Sigma} = 0, \qquad \left.\d_{(i)} \left(\dot z^A -
	\sigma^{AB}\frac{\delta H}{\delta z^B}\right)\right\vert_{\d\Sigma}=0.
\end{equation}
This requirement is stronger than the one imposed in section
\ref{sec:gendifffunc} as the smoothness conditions are automatically satisfied
when the EOM are imposed on the boundary. Imposing all equations of motion on
the boundary may seem like a very strong requirement, however, as we will only
be concerned by symmetries, it will not restrict our analysis.
The conditions \eqref{eq:gaugeBCEOM} will be referred as the external boundary
conditions.
We will also assume the same regularity requirement as in section
\ref{sec:gendifffunc} for
the variations: the set of variations preserving the internal boundary
conditions
\begin{equation}
	\delta \chi_E\vert_{\d\Sigma} = 0, \quad \delta
	\chi_I^\alpha\vert_{\d\Sigma} = 0 \qquad \forall (z, \lambda) \text{
	s.t. } \chi_E\vert_{\d\Sigma} = 0,
\end{equation}
when evaluated on a specific value of $\zeta^\alpha$ is equal to the set
\begin{equation}
	\delta \chi_E\vert_{\d\Sigma} = 0, \quad \delta
	\chi_I^\alpha\vert_{\d\Sigma} = 0 \qquad \forall (z, \lambda) \text{
	s.t. } \chi_E\vert_{\d\Sigma} = 0, \quad \chi_I^\alpha
	\vert_{\d\Sigma} = \zeta^\alpha.
\end{equation}
We will require the total hamiltonian $H$ to have a well defined variation: for all variations $\delta$ preserving
the internal boundary conditions, we have
\begin{gather}
	\delta H = \int_\Sigma d^nx \left(\frac{\delta H}{\delta z^A}\delta
	z^A + \phi_a \delta \lambda^a\right).
\end{gather}

\vspace{5mm}

Let's consider a couple $(G, \delta^\lambda_G)$ such that
\begin{gather}
	\label{eq:gaugecoupleI}
	G[t,z,,\lambda, \dot \lambda,
	...,\overset{(k)}{\lambda}],\qquad \delta^\lambda_G \lambda^a = \Lambda^a(t,z,\lambda, \dot \lambda,
	...,\overset{(k)}{\lambda}),\\
	\label{eq:gaugecoupleII}
	[\delta_G, \delta] \chi_I^\alpha \vert_{\d\Sigma} = 0,\quad
	\delta G = \int_\Sigma d^nx \left(\frac{\delta G}{\delta z^A}\delta
		z^A + \sum_{l=0} \frac{\delta G}{\delta
	\overset{(l)}{\lambda^a}} \delta \overset{(l)}{\lambda^a}\right),
\end{gather}
for all variation $\delta$ preserving the internal boundary conditions. 
Using arguments similar to those used in section \ref{sec:gendifffunc}, one can prove:
\begin{theorem}
	A couple $(G, \delta^\lambda_G)$ of the form \eqref{eq:gaugecoupleI}
	satisfying \eqref{eq:gaugecoupleII} is an external symmetry of the
	action if and only if
there exists a boundary term $\oint V_G$ such that:
\begin{gather}
	\label{eq:GFTconscond}
	\frac{D}{Dt}G + \{G, H\}_\zeta- \delta^\lambda_GH = \oint_{\d\Sigma} (d^{n-1}x)_i V^i_G(z, \lambda, \dot \lambda,
	...),\\
	Z^A = \sigma^{AB}\frac{\delta G}{\delta z^B},\quad \frac{D}{Dt} =
	\frac{\d}{\d t} + \sum_{l=0} \d_{(i)}\overset{(l+1)}{\lambda^a}
	\frac{\d}{\d \overset{(l)}{\lambda^a}_{(i)}},
\end{gather}
with 
\begin{gather}
	\delta\oint_{\d\Sigma} (d^{n-1}x)_i V^i_G(z, \lambda, \dot \lambda,
	...)=0,
\end{gather}
for all variations $\delta$ preserving the internal boundary conditions.
\end{theorem}
\noindent Remark that, because the transformations we are studying are locally
symmetries of the equations of motion, the external boundary conditions are
always preserved. This theorem has an interesting corollary which is the
field theory
equivalent of equation \eqref{eq:symmdeplambda}:
\begin{corollary}
	\label{theo:coroll}
	If a couple $(G, \delta^\lambda_G)$ of the form \eqref{eq:gaugecoupleI}
	satisfying \eqref{eq:gaugecoupleII} is an external symmetry of the
	action then
	\begin{equation}
		\label{eq:GFTsymmdeplambda}
		\frac{\delta G}{\delta \overset{(l)}{\lambda^a}}\approx 0,
		\qquad \forall l.
	\end{equation}
\end{corollary}
\begin{proof}
	Using
	\begin{equation}
		\label{eq:GFTdefdelt}
		\delta_t = \frac{D}{Dt} +
		\d_{(i)}\left(\sigma^{AB}\frac{\delta H}{\delta z^B}\right)
		\frac{\d^S}{\d z_{(i)}^A} = \frac{\d}{\d t} + \delta_H, \quad
		\delta_H \lambda^a = \dot \lambda^a,
	\end{equation}
	the conservation equation \eqref{eq:GFTconscond} can be written
	\begin{equation}
		\delta_tG - \delta_G^\lambda H = \oint_{\d\Sigma} V_G.
	\end{equation}
	For $l>0$, if $\delta^l$ is an arbitrary bulk variation of
	$\overset{(l)}{\lambda^a}$ only, we get
	\begin{equation}
		\delta_t \delta^l G + \delta^l \overset{(l)}{\lambda^a}
		\frac{\delta G}{\delta \overset{(l-1)}{\lambda^a}} - \delta^l
		\Lambda^a \phi_a = 0,
	\end{equation}
	up to a boundary term. Because $\delta_t \phi_a \approx 0$, 
	starting from $l=k+1$, we obtain all the identities \eqref{eq:GFTsymmdeplambda} recursively.
\end{proof}
\noindent The main difference compared to the previous section is that in this case, we cannot remove the dependence of
$G$ in the lagrange multipliers due to the boundary conditions involving both
$\lambda$ and $z$.

\vspace{5mm}

Let's now consider two couples $(G_1, \delta_1^\lambda)$ and $(G_2,
\delta_2^\lambda)$ generating external symmetries of the action. The
commutator of these two symmetries will contain dependences in the time
derivatives of the dynamical variables $z^A$. As in section \ref{sec:symgaugetheo}, we can
remove them using trivial symmetry transformations to obtain $\delta_{[1,2]} \sim [\delta_1,
\delta_2]$ with:
\begin{eqnarray}
	\delta_{[1,2]}\lambda^a & = & \tilde \delta_1 \Lambda_2^a - \tilde
	\delta_2 \Lambda_1^a,\\
	\delta_{[1,2]}z^A & = & \tilde \delta_1 Z^A_2\nonumber\\
	&& +
	\sum_{l=1}\sum_{m=0}^{l-1} \sigma^{AD}\frac{\delta}{\delta z^D_{(j)}}
	(\delta_t^m \Lambda_1^b) (-\d)_{(j)}
	(-\delta_t)^{l-1-m}\left(\frac{\delta
	}{\delta \overset{(l)}{\lambda^b}}(\phi_a \Lambda^a_2)\right)\nonumber\\
	&& \quad - (1\leftrightarrow 2),
\end{eqnarray}
where
\begin{equation}
	\tilde \delta_G = \d_{(i)} \left(\sigma^{AB}\frac{\delta G}{\delta
	z^B}\right) \frac{\d^S}{\d z^A_{(i)}} + \sum_{l=0} \d_{(i)}
		\left(\delta^l_t \Lambda^a_G\right) \frac{\d^S}{\d
		\overset{(l)}{\lambda^a}_{(i)}},
\end{equation}
with $\delta_t$ defined in equation \eqref{eq:GFTdefdelt}. Because $\delta_t -
d_t$ is proportional to the EOM and we imposed them
on the boundary, the transformations $\tilde \delta_1$, $\tilde \delta_2$ and
$\delta_{[1,2]}$ all preserve the external boundary conditions. One can check
that the dynamical part of the transformation $\delta_{[1,2]}$ is generated by
\begin{equation}
	G_{[1,2]} = \{G_2, G_1\}_\zeta + \tilde\delta^\lambda_1 G_2 - \tilde
	\delta^\lambda_2G_1 = \tilde\delta_1 G_2 - \tilde
	\delta_2G_1 - \{G_2, G_1\},
\end{equation}
where $\{, \}_\zeta$ is the bracket involving only $z^A$ defined in section
\ref{sec:gendifffunc}.

\begin{theorem}
	If $(G_1, \delta_1^\lambda)$ and $(G_2,
\delta_2^\lambda)$ generate external symmetries of the action, then their
bracket $[(G_1, \delta_1^\lambda),(G_2,
\delta_2^\lambda)]$ defined by
\begin{gather}
	\label{eq:GFTbracket}
[(G_1, \delta_1^\lambda),(G_2,
\delta_2^\lambda)] = (\{G_2, G_1\}_g, \delta_{[1,2]}^\lambda),\\
\{G_2, G_1\}_g = \tilde\delta_1 G_2 - \tilde
	\delta_2G_1 - \{G_2, G_1\},\\
	\delta^\lambda_{[1,2]} \lambda^a = \tilde \delta_1 \Lambda_2^a - \tilde
	\delta_2 \Lambda_1^a,
\end{gather}
satisfies \eqref{eq:gaugecoupleII} for all variations preserving the internal
boundary conditions and generates an external symmetry of the action.
\end{theorem}
\noindent Using this bracket, the conservation condition \eqref{eq:GFTconscond} can be
rewritten
\begin{equation}
	\label{eq:GFTconscondII}
	\frac{\d}{\d t} +\{G, H\}_g = \oint_{\d\Sigma}(d^{n-1})_i V^i_G.
\end{equation}
\begin{proof}
	The first condition of \eqref{eq:gaugecoupleII} is direct using
	the Jacobi identity for the variations. In appendix \ref{sec:diffbrack}, we prove that for all variations $\delta$ preserving the internal
boundary conditions, we have:
\begin{multline}
	\delta \{G_2, G_1\}_g = \delta z^A \frac{\delta}{\delta z^A} \{G_2,
	G_1\}_g\\ + \sum_{l,m=0} \delta \overset{(l)}{\lambda^a} (-\d)_{(i)}
	\left(\frac{\d^S \delta^m_t\Lambda^b_1}{\d
		\overset{(l)}{\lambda^a}_{(i)}}\frac{\delta G_2}{\delta
	\overset{(m)}{\lambda^b}}-\frac{\d^S \delta^m_t\Lambda^b_2}{\d
		\overset{(l)}{\lambda^a}_{(i)}}\frac{\delta G_1}{\delta
	\overset{(m)}{\lambda^b}}\right),
\end{multline}
which is exactly equation \eqref{eq:gaugecoupleII}.

In appendix \ref{sec:gaugemodjacobi}, we prove that
\begin{equation}
	\{\{G_2, G_1\}_g, H\}_g = -\{\{G_2, H\}_g,G_1\}_g +\{\{G_1,H \}_g,
	G_2\}_g.
\end{equation}
Because both $(G_1, \delta_1^\lambda)$ and $(G_2, \delta_2^\lambda)$ generate
symmetries of the action, they satisfy the conservation condition \eqref{eq:GFTconscondII}. Combining these identities with
$\delta^\lambda_{[1,H]}=0$, we get
\begin{equation}
	\frac{\d}{\d t} \{G_2, G_1\}_g + \{\{G_2, G_1\}_g, H\}_g =
	\oint_{\d\Sigma} (\tilde \delta_2V_1-\tilde \delta_1V_2),
\end{equation}
with 
\begin{equation}
	\delta \oint_{\d\Sigma} (\tilde \delta_2V_1-\tilde \delta_1V_2)= 0
\end{equation}
for all $\delta$ preserving the internal boundary conditions. This proves that
the bracket of two external symmetries is an external symmetry.
\end{proof}

A subset of the symmetries of the action are proper gauge transformations.
These transformations generate the redundancy in the description of the
theory. We will define them as:
\begin{definition}
	The transformation $\delta_\Gamma$ generated by a couple 
	$(\Gamma, \delta^\lambda_\Gamma)$ is a proper gauge transformation if
	it is an external symmetry
	preserving the internal boundary conditions
	$\delta_\Gamma\chi^\alpha_I\vert_{\d\Sigma} = 0$ and if its generator satisfies:
	\begin{equation}
		\Gamma \approx 0.
	\end{equation}
\end{definition}
\noindent The requirement here is stronger than usual. The main difference is that we
are treating all values of the boundary field $\zeta^\alpha$ at the same time.
It is possible that the set of proper gauge transformation defined above, when
evaluated for a specific value of $\zeta^\alpha$, is smaller than the set
computed at fixed $\zeta^\alpha$ \cite{troessaert_canonical_2013}. In the following, we will assume that the
set of proper gauge transformations defined here generates all proper gauge
transformation for each value of the boundary field $\zeta^\alpha$. In other
word, fixing these proper gauge transformations completely fixes the gauge
freedom of each independent theory associated with the different values of the
boundary conditions.

With the above definition, we obtain the expected result:
\begin{theorem}
	If the transformation generated by a couple $(G, \delta^\lambda_G)$ is
	an external symmetry then $G$ is a first-class functional:
	for all $(\Gamma, \delta^\lambda_\Gamma)$ generating proper gauge
	transformations, the external symmetry generated by $[(G,
	\delta^\lambda_G),(\Gamma,\delta^\lambda_\Gamma)]$ is a proper gauge
	transformation.
\end{theorem}
\begin{proof}
	Let's consider two couples $(G, \delta^\lambda_G)$ and $(\Gamma,
	\delta^\lambda_\Gamma)$ respectively generating an external symmetry
	and a proper gauge symmetry.
	By definition, as $\delta_\Gamma$ preserves the internal boundary
	conditions, the bracket $[\delta_G, \delta_\Gamma]$ will also preserve
	them.
	We have $\Gamma \approx 0$ for all values of $\zeta^\alpha$: this
	implies $\delta_G\Gamma\approx 0$. As $\delta_\Gamma$ preserves the
	internal boundary conditions, we also have:
	\begin{equation}
		\delta_\Gamma G - \{G, \Gamma\} = \int_\Sigma d^nx \,
		\sum_{l=0}\frac{\delta G}{\delta \overset{(l)}{\lambda^a}} 
		\delta \overset{(l)}{\lambda^a} \approx 0,
	\end{equation}
	using corollary \ref{theo:coroll}. These two results combine to 
	\begin{equation}
		\{G,\Gamma\}_g \approx 0.
	\end{equation}
\end{proof}
\noindent Another way of expressing this result is that the sub-algebra of proper gauge
transformations forms an ideal. The Jacobi identity for the modified bracket
defined in \eqref{eq:GFTbracket} is not valid (see appendix \ref{sec:gaugemodjacobi}). However, we
still have
\begin{theorem}
	The bracket induced on the quotient of the couples $(G,
	\delta^\lambda_G)$ generating external symmetries by the couples
	generating proper gauge transformations forms a representation of the
	algebra obtained form the quotient of external symmetries by proper
	gauge transformations.
\end{theorem}
\begin{proof}
	Let's consider three couples $(G_1, \delta_1^\lambda)$, $(G_2,
	\delta_2^\lambda)$ and $(G_3, \delta_3^\lambda)$ generating external
	symmetries. We proved in appendix
	\ref{sec:gaugemodjacobi} that the cyclic combination 
	\begin{equation}
		[[(G_1, \delta_1^\lambda), (G_2, \delta_2^\lambda)],(G_3,
		\delta_3^\lambda)] + cyclic = (G_J, \delta_J^\lambda)
	\end{equation}
	satisfies  $G_J \approx 0$. Due to the Jacobi
	identity of transformations, we know that $\delta_J$ differs form
	zero by a trivial symmetry. Because we imposed all equations of motion
	as boundary conditions, the transformation $\delta_J$ preserves the
	internal boundary conditions.
	This proves that $(G_J, \delta^\lambda_J)$ generates a proper gauge
	transformation. From this, the theorem follows easily.
\end{proof}

Let's consider two couples $(G_1, \delta_1^\lambda)$ and $(G_2,
\delta_2^\lambda)$ generating external symmetries as well as
$(G_{[1,2]},\delta_{[1,2]}^\lambda)$ a generator of the external symmetry
$\delta_{[1,2]} \sim [\delta_1, \delta_2]$. We showed that $(\{G_2,G_1\}_g,
\delta^\lambda_{[1,2]})$ also generates $\delta_{[1,2]}$. This means that both
functional differ by a functional that is in the kernel of
$\frac{\delta}{\delta z^A}$:
\begin{gather}
	\{G_2, G_1\} = G_{[1,2]} + K_{1,2},\\
	K_{1,2} = \int_\Sigma d^nx \, K^{bulk}_{1,2}(t,\lambda, ...) + \oint_\Sigma
	(d^{n-1}k)_i\, K^i_{1,2}(t,z,\lambda, ...).
\end{gather}
Because of equation \eqref{eq:GFTsymmdeplambda}, the bulk part of the functional must be 
independent of the lagrange multipliers. What is left can be absorbed into the
boundary term. Equation \eqref{eq:gaugecoupleII} then implies that $K_{1,2}$
can only depend on the boundary fields $\zeta$:
\begin{equation}
	\delta K_{1,2} = \delta \oint_{\d\Sigma} (d^{n-1}x)_i K^i_{1,2}
	= 0,
\end{equation}
for all variations $\delta$ preserving the internal boundary conditions. Using
the previous results, we also get
\begin{theorem}
If $G_i$ forms a
generating set of the algebra $\mathcal{G}$, we have in general
\begin{equation}
	[(G_1, \delta^\lambda_1), (G_2, \delta^\lambda_2)] = \left(G_{[1,2]}
		+ K_{1,2}, 
	\delta^\lambda_{[1,2]}\right),
\end{equation}
with $K_{1,2}$ antisymmetric and
\begin{equation}
K_{[1,2],3} + \delta^j_3 K_{1,2} + cyclic = 0.
\end{equation}
\end{theorem}
\noindent As in the previous cases, the representation
of the algebra of external symmetries has room for abelian extensions.

\subsection{Link with the non-integrability of surface charge}
\label{sec:nonintegsurf}

The method usually used to define charges for improper gauge transformations is the
following (see \cite{troessaert_canonical_2013} and references therein): after having
chosen a set of boundary conditions for which the action is well defined, one
computes the set of gauge transformations preserving the boundary conditions
and then select the subset for which it is possible to define differentiable
generators. In order to make the link with the existing literature easier, 
we will only consider transformations of the
canonical variables that are independent of the lagrange multipliers. 

The integrability problem appears in the last step. One has to solve an integrability condition
to find the correct boundary term $k_\epsilon$:
\begin{gather}
	\delta \int_\Sigma d^nx\epsilon^a \phi_a = \int_\Sigma d^nx \,
	\frac{\delta \epsilon^a\phi_a}{\delta z^A}\delta z^A +
	\oint_{\d\Sigma} (d^{n-1}x)_i \Theta^i_\epsilon[\delta z],\\
	\label{eq:integintcond}
	\delta\oint_{\d\Sigma}(d^{n-1}x)_i k_\epsilon^i = -
	\oint_{\d\Sigma}(d^{n-1}x)_i \Theta^i_\epsilon[\delta z],
\end{gather}
where the equalities are valid for all variations $\delta$ preserving the boundary conditions.
If such a boundary term exists, the differentiable generator of the
transformation $\delta_\epsilon$ associated to the gauge parameter $\epsilon$
is then
\begin{equation}
	G_\epsilon = \int_\Sigma d^nx\epsilon^a \phi_a +
	\oint_{\d\Sigma}(d^{n-1}x)_i k_\epsilon^i.
\end{equation}
If this boundary term does not exists, the associated transformation is not
canonical and it cannot be regarded as a symmetry of the action. 

The usual way to solve this problem is to tighten the restrictions on the boundary.
This reduces the set of variations in \eqref{eq:integintcond} and may help to define a
suitable boundary term. However, doing this also reduces the
set of gauge transformations preserving the boundary conditions and perhaps
remove the transformations of interest. 

The notion of external
symmetry we introduced in this work brings another solution to this problem.
The main idea is that the boundary conditions on the symmetries and the
boundary conditions on the transformations used in the integrability
conditions are different. Symmetries must only preserve the external boundary
conditions and transform the boundary fields of the internal boundary
conditions in an appropriate way whereas the transformations used in the
variation condition \eqref{eq:gaugecoupleII} must preserve the stronger internal boundary conditions.


\section*{Acknowledgements}
\label{sec:acknowledgements}

\addcontentsline{toc}{section}{Acknowledgments}

I would like to thank G. Barnich and P. Ritter for
useful discussions. This work is founded by the fundecyt postdoctoral grant
3140125. The Centro de Estudios Cient\'ificos (CECs) is funded by the Chilean
Government through the Centers of Excellence Base Financing Program of
Conicyt.

\appendix

\section{Algebra for the poisson bracket with sources}
\label{sec:appsources}
In this appendix, we will show that $\left(\{G,F\}_j, [ \delta_F^j,
  \delta_G^j]\right)$ satisfies
\begin{equation}
\partial_t \{G,F\}_j + \left\{\{G,F\}_j,H \right\}_j=V_{\{G,F\}}.
\end{equation}

We start by using the Jacobi identity on the second term:
\begin{eqnarray}
\left\{\{G,F\}_j,H \right\}_j & = & \left\{\{H,F\}_j,G
\right\}_j+\left\{F,\{H,G\}_j \right\}_j\\
 & =&  \left\{\partial_t F - V_F,G
\right\}_j+\left\{F,\partial_t G -V_G\right \}_j\label{eq:app1}
\end{eqnarray}
Using the fact that both $(F, \delta^j_F)$ and $(G, \delta^j_G)$ are
generators of symmetry. Using equation (\ref{eq:commut}), we see that
$\partial_t F-V_F$ is associated to the following operator
acting on the sources:
\begin{eqnarray}
\delta_{\partial_t F}^j j^\alpha& = & \left[ \delta^j_F,\delta^j_H
\right]j^\alpha\\
&=& \delta_F^j d_t j^\alpha-\delta^j_t J^\alpha_F \\
&=& \partial_t J^\alpha_F.
\end{eqnarray}
Now, we can expend equation (\ref{eq:app1}) to
\begin{eqnarray}
\left\{\{G,F\}_j,H \right\}_j & =&  \left\{\partial_t F - V_F ,G
\right\}+ \delta^j_G \left(\partial_t F -V_F\right)- \delta^j_{\partial_t F}
G\nonumber \\ && \quad +\left\{F,\partial_t G -V_G\right\} + \delta^j_{\partial_t G} F-
\delta^j_F \left(\partial_t G -V_G \right) \nonumber\\
& = & \partial_t \left\{F,G
\right\}+ \delta^j_G \left(\partial_t F -V_F\right) + \delta^j_{\partial_t G} F-
\delta^j_F \left(\partial_t G-V_G\right) - \delta^j_{\partial_t F}
G\nonumber\\
& = & \partial_t \left\{F,G
\right\}+ \partial_t\delta^j_G F -
\partial_t\delta^j_F G - \delta^j_G V_F + \delta^j_F V_G\nonumber\\
&=& \partial_t \left\{F,G \right\}_j + \left(\delta^j_F V_G- \delta^j_G V_F \right),
\end{eqnarray}
which is what we wanted with $V_{\{G,F\}} = \delta^j_F V_G- \delta^j_G V_F$. Between the second and the third line, we
used the following identity 
\begin{equation}
\left[\partial_t, \delta_G^j \right] = \left[d_t - \delta^j_t -
  \dot z^A \frac{\partial}{\partial z^A} , \delta_G^j\right]
=  \left[- \delta^j_t, \delta_G^j\right]
=\delta^j_{\d_t G}.
\end{equation}

\section{Variation of the modified brackets}
\label{sec:diffbrack}

In this appendix, we will study the behavior of our new brackets for field
theories. This analysis applies to both the poisson bracket taking into
account the boundary conditions defined in section \ref{sec:bound-cond-as} and the bracket of
conserved quantities for gauge field theories defined in section \ref{sec:GFT}.

Both bracket can be written as
\begin{equation}
	\left\{F, G\right\}^{mod} = \tilde\delta_G F - \tilde\delta_F G - \left\{F,
	G\right\},
\end{equation}
where the variations act on all the fields. The definition of
$\tilde\delta_{F,G}$ is given in section \ref{sec:GFT}. In the non-gauge theory case, will assume that $F$ and $G$
are generalized differentiable functionals. In the gauge theory case, we will
consider two couples $(F, \delta^\lambda_F)$ and $(G, \delta^\lambda_G)$
satisfying \eqref{eq:gaugecoupleII} and generating external symmetries of the action. If $\delta$ is a variation
preserving the internal boundary conditions then $[\delta, \tilde\delta_G]$ also preserves
them. We have
\begin{eqnarray}
	\delta\tilde\delta_G F & = & [\delta, \tilde\delta_G]F + \tilde\delta_G
	\delta
	F\nonumber\\
	& = & \delta Z_G^A \frac{\delta F}{\delta z^A} + \sum_{l=0}\delta
	\delta_t^l\Lambda_G^a
	\frac{\delta F}{\delta \overset{(l)}{\lambda^a}} + \delta z^A \tilde\delta_G \frac{\delta
	F}{\delta z^A} + \sum_{l=0}\delta \overset{(l)}{\lambda^a} \tilde\delta_G\frac{\delta F}{\delta
	\overset{(l)}{\lambda^a}}.
\end{eqnarray}
In the non-gauge theory case, the terms containing $\lambda$ are absent.
Commuting $\tilde\delta_G$ and $\frac{\delta}{\delta z^A}$, we get
\begin{equation}
	\tilde\delta_G \frac{\delta F}{\delta z^A} = \frac{\delta}{\delta
	z^A}\tilde\delta_G F - \delta z^A (-\d_i) \left(\frac{\d Z^B_G}{\d
		z^A_{(i)}} \frac{\delta F}{\delta z^B} + \sum_{l=0} \frac{\d
		\delta_t^i \Lambda^b_G}{\d z^A_{(i)}} \frac{\delta F}{\delta
	\overset{(l)}{\lambda^b}}\right).
\end{equation}
Using the fact that we imposed the constraints and all their derivatives to be
zero on the boundary and corollary \ref{theo:coroll}, we also have
\begin{equation}
	\delta \delta^l_t \Lambda^a_G \frac{\delta F}{\delta
	\overset{(l)}{\lambda^a}} = \delta z^A (-\d_{(i)}) \left(\frac{\d
		\delta^l_t \Lambda^a_G}{\d z^A_{(i)}} \frac{\delta F}{\delta
	\overset{(l)}{\lambda^a}}\right)+ \sum_{m=0} \delta
	\overset{(m)}{\lambda^b} (-\d_{(i)}) \left(\frac{\d
		\delta^l_t \Lambda^a_G}{\d \overset{(m)}{\lambda^b_{(i)}}} \frac{\delta F}{\delta
	\overset{(l)}{\lambda^a}}\right).
\end{equation}
Lastly, the
variation of the non-modified poisson bracket can be written as:
\begin{equation}
	\delta \{F, G\} = \delta Z^A_G \frac{\delta F}{\delta z^A} - \delta
	Z^A_F \frac{\delta G}{\delta z^A}.
\end{equation}
Combining everything, we get
\begin{multline}
	\delta \{F, G\}^{mod} = \delta z^A \frac{\delta}{\delta z^A}
	\left( \tilde \delta_GF - \tilde \delta_FG - \{F,G\}\right)\\ 
	\qquad + \sum_{l,m=0} \delta
	\overset{(m)}{\lambda^b} (-\d)_{(i)} \left(\frac{\d
		\delta^l_t \Lambda^a_G}{\d \overset{(m)}{\lambda^b_{(i)}}} \frac{\delta F}{\delta
	\overset{(l)}{\lambda^a}} - \frac{\d
		\delta^l_t \Lambda^a_F}{\d \overset{(m)}{\lambda^b_{(i)}}} \frac{\delta G}{\delta
	\overset{(l)}{\lambda^a}}\right).
\end{multline}

\section{Jacobi identity of the modified bracket}
\label{sec:jacobiMB}

In this appendix, we will use objects of the bi-variational
formalism. We will follow the definitions, notations and properties introduced in 
appendix A of \cite{barnich_surface_2008}.

To prove the Jacobi identity, we will use the following expression
 for the modified poisson bracket of the n-forms $\hat F$ and $\hat G$:
\begin{equation}
\left \{ \hat F, \hat G\right \}_{\zeta}= \delta_G \hat F  - d_H I^n_F \hat G.
\end{equation}
The corresponding bracket of the associated functionals $F = \int \hat F$ and
$G = \int \hat G$ is
\begin{equation}
\left \{ F, G\right \}_{\zeta} = \int_\Sigma\left \{ \hat F, \hat G\right
\}_{\zeta}.
\end{equation}
We can rewrite the Jacobi identity as
\begin{eqnarray}
\left\{\hat F,\left\{\hat G,\hat J\right\}_{\zeta} \right\}_{\zeta} + cycl
&=&\left\{\left\{ \hat J, \hat G \right\} _{\zeta},\hat  F \right\}_{\zeta} -\left\{\left\{
   \hat  J,\hat F\right\}_{\zeta} ,\hat  G\right\}_{\zeta}\nonumber \\ && \quad
-\left\{ \hat J,\left\{\hat  G ,\hat F \right\}_{\zeta}
\right\}_{\zeta}\\ &=&\delta_F\left\{\hat  J ,\hat G \right\}_{\zeta} -d_H I^n_{\{J ,G\}}
\hat F \nonumber \\ && \quad-\delta_G\left\{\hat J,\hat F \right\}_{\zeta} + d_H I^n_{\{ J ,F \}}
\hat G\nonumber \\ && \quad - \delta_{\{G,F\}} \hat J + d_H I^n_J
\left\{\hat  G,\hat F \right\}_{\zeta}\\ 
&=&\delta_F\left( \delta_G \hat J  -d_H I^n_J \hat G \right) -d_H I^n_{\{ J ,G \}}
\hat F\nonumber \\ && \quad-\delta_G\left( \delta_F \hat J  -d_H I^n_J \hat F\right) + d_H I^n_{\{ J ,F \}}
\hat G\nonumber \\ && \quad - \delta_{\{G,F\}} \hat J + d_H I^n_J \left( \delta_F \hat G -d_H I^n_G \hat F\right)\nonumber\\
&=&(\delta_F \delta_G -\delta_G\delta_F  - \delta_{\{G,F\}}  )\hat J 
\nonumber \\ && \quad+d_H\left\{- \delta_F I^n_J \hat G- I^n_{\{ J,G \}}
\hat F +\delta_G I^n_J \hat F\right.\nonumber\\ && \qquad \left.+  I^n_{\{ J ,F \}}
\hat G+  I^n_J \left( \delta_F \hat G -d_H I^n_G \hat F\right) \right\}.
\end{eqnarray}
From (\ref{eq:commutvar}), we see that the first line gives zero, we are left with the boundary
term. It can be simplified to
\begin{eqnarray}
Jacobi & = & d_H\left\{- \delta_F I^n_J \hat G- I^n_{[ G,J ]}
\hat F +\delta_G I^n_J \hat F\right.\nonumber\\ && \qquad \left.+  I^n_{[ F,J ]}
\hat G +  I^n_J \delta_F \hat G  - \delta_J I^n_G \hat F \right\}\\
& = & d_H\left\{- \left[\delta_F, I^n_J\right] \hat G  +  I^n_{[ F,J ]}
\hat G \right.\nonumber\\ && \qquad +\left[\delta_G, I^n_J\right] \hat F  - I^n_{[ G,J ]}
\hat F \nonumber\\ && \qquad \left. + I^n_J \delta_G \hat F- \delta_J I^n_G \hat F\right\}.
\end{eqnarray}
Applying twice equations (A.40) of \cite{barnich_surface_2008}, we obtain
\begin{equation}
	\label{eq:appjacobiIII}
Jacobi  =  d_H\left\{T_F \left[J, \frac{\delta \hat G }{\delta z}
  \right] -T_G \left[J, \frac{\delta \hat F}{\delta z} \right]+ I^n_J
  \delta_G \hat F - \delta_J I^n_G \hat F \right\}.
\end{equation}
We have two useful properties of these $T$:
\begin{eqnarray}
\label{propT1-BH}
T_{Q_1} \left[Q_2, \frac{\delta \omega^n}{\delta z}
  \right]&=& -W_{\frac{\delta \omega^n}{\delta z}} \left[Q_2,Q_1
  \right] + I^n_{Q_2} \left( Q_1^B \frac{\delta \omega^n}{\delta
      z^B}\right),\\
\label{propT2-BH}
T_F \left[G, \frac{\delta \hat J}{\delta z}
  \right]&=&-W_{\frac{\delta \hat F}{\delta z}} \left[G,J
  \right].
\end{eqnarray}
The first one is the equation (A.52) of \cite{barnich_surface_2008}. The
second one is just coming from the definition of $T$ and the
properties of hamiltonian
generators:
\begin{eqnarray}
T_F \left[G, \frac{\delta \hat J}{\delta z}
  \right]&=& \left(\begin{array}{c} \vert \mu \vert + \vert \rho \vert
      + 1 \\ \vert \mu \vert + 1 \end{array}\right) \partial_{(\mu)}
  \left(G^A (-\partial)_{(\rho)} \left( \frac{\partial^S F^B}{\partial
      z^A_{(\mu)(\rho)\nu}} \frac{\partial}{\partial d x^\nu}
    \frac{\delta \hat J}{ \delta z^B}\right)\right)\nonumber\\
&=& \left(\begin{array}{c} \vert \mu \vert + \vert \rho \vert
      + 1 \\ \vert \mu \vert + 1 \end{array}\right) \partial_{(\mu)}
  \left(G^A (-\partial)_{(\rho)} \left( \sigma^{BC}\frac{\partial^S }{\partial
      z^A_{(\mu)(\rho)\nu}} \frac{\delta F }{ \delta z^C}
    \frac{\delta J }{ \delta z^B} d^{n-1}_\nu
    x\right)\right)\nonumber\\
&=& -\left(\begin{array}{c} \vert \mu \vert + \vert \rho \vert
      + 1 \\ \vert \mu \vert + 1 \end{array}\right) \partial_{(\mu)}
  \left(G^A (-\partial)_{(\rho)} \left( \sigma^{CB}\frac{\delta J }{ \delta z^B} \frac{\partial^S }{\partial
      z^A_{(\mu)(\rho)\nu}} \frac{\partial}{\partial d
      x^\nu}\frac{\delta (Fd^nx) }{ \delta
      z^C}\right)\right)\nonumber\\
&=& -\left(\begin{array}{c} \vert \mu \vert + \vert \rho \vert
      + 1 \\ \vert \mu \vert + 1 \end{array}\right) \partial_{(\mu)}
  \left(G^A (-\partial)_{(\rho)} \left( J^C\frac{\partial^S }{\partial
      z^A_{(\mu)(\rho)\nu}} \frac{\partial}{\partial d x^\nu}\frac{\delta \hat F }{ \delta z^C}\right)\right)\nonumber\\
&=&-W_{\frac{\delta \hat F}{\delta z}} \left[G,J
  \right],
\end{eqnarray}
using equation (A.51) of \cite{barnich_surface_2008}. The next step is to
apply (\ref{propT1-BH}) to the second term of \eqref{eq:appjacobiIII} 
and  (\ref{propT2-BH}) to the first term :
\begin{eqnarray}
Jacobi  & =  & d_H\left\{-W_{\frac{\delta \hat F }{\delta z}} \left[J,G
  \right]+W_{\frac{\delta \hat F}{\delta z}} \left[J,G
  \right] -I_J^n \left( G^A \frac{\delta \hat F }{\delta z^A}\right)
+ I^n_J \delta_G \hat F- \delta_J I^n_G \hat F\right\}\nonumber\\
& =  & d_H\left\{-I_J^n \left( \delta_G \hat F  -d_H I_G \hat F\right)
+ I^n_J \delta_G \hat F- \delta_J I^n_G \hat F\right\}\nonumber\\
&=&0.
\end{eqnarray}

\section{Jacobi identity for gauge field theories}
\label{sec:gaugemodjacobi}

This appendix is devoted to the computation of the cyclic identity of the
modified bracket of conserved quantities in section \ref{sec:GFT}. The form we
will use is
\begin{equation}
	\{F, G\}_g = \{F, G\}_\zeta +
	\tilde\delta^\lambda_GF-\tilde\delta^\lambda_FG.
\end{equation}

Let's first prove a useful identity:
\begin{multline}
	\label{eq:applemma}
		\tilde \delta^\lambda\{F, G\}_\zeta = \{\tilde \delta^\lambda
		F, G\}_\zeta + \{F, \tilde \delta^\lambda G\}_\zeta
		+ \sum_{l=0}\d_{(i)} \left[\delta^z_F(\delta_t^l\Lambda^a)
			\frac{\delta G}{\delta
		\overset{(l)}{\lambda^a}_{(i)}}-\delta^z_G(\delta_t^l\Lambda^a)
			\frac{\delta F}{\delta
		\overset{(l)}{\lambda^a}_{(i)}}\right],
\end{multline}
where $\delta^\lambda \lambda^a = \Lambda^a$.
\begin{proof}
	We can write
	\begin{equation}
		\{F, G\}_\zeta= \delta^z_GF - \delta^z_FG -\{F, G\}.
	\end{equation}
	Using
	\begin{eqnarray}
		\tilde\delta^\lambda\delta^z_GF &=& \delta^z_G\tilde\delta^\lambda F +
		\d_{(i)} \left[\tilde \delta^\lambda(\sigma^{AB}\frac{\delta
			G}{\delta z^B}) \frac{\delta F}{\delta z^A_{(i)}} +
				\sum_{l=0}(-\delta^z_G\delta_t^l\Lambda^a)
				\frac{\delta F}{\delta
			\overset{(l)}{\lambda^a}_{(i)}}\right]\nonumber\\
		&=& \delta^z_G\tilde\delta^\lambda F +
		\d_{(i)} \left[\sigma^{AB}\frac{\delta
			}{\delta z^B}(\tilde\delta^\lambda G) \frac{\delta F}{\delta
		z^A_{(i)}}-\sum_{l=0}\delta^z_G(\delta_t^l\Lambda^a)
			\frac{\delta F}{\delta
			\overset{(l)}{\lambda^a}_{(i)}}\right] \\ && 
			\quad +
			\d_{(i)} \left[-\sum_{l=0}\sigma^{AB}(-\d)_{(j)}
				\left(\frac{\d \delta^l_t\Lambda^a}{\d
					z^B_{(j)}}\frac{\delta G}{\delta
				\overset{(l)}{\lambda^a}}\right)
			\frac{\delta F}{\delta z^A_{(i)}}\right],\nonumber
	\end{eqnarray}
	with
	\begin{equation}
		\tilde \delta^\lambda\frac{\delta F}{\delta z^A}
		\sigma^{AB}\frac{\delta G}{\delta z^B} = 
		\frac{\delta }{\delta z^A}(\tilde \delta^\lambda F)
		\sigma^{AB}\frac{\delta G}{\delta z^B}- \sum_{l=0} (-\d)_{(i)}
		\left(\frac{\d \delta_t^l \Lambda^a}{\d z^A_{(i)}}
			\frac{\delta F}{\delta
		\overset{(l)}{\lambda^a}}\right) \sigma^{AB}\frac{\delta
		G}{\delta z^B},
	\end{equation}
	we get
	\begin{eqnarray}
		\tilde \delta^\lambda\{F, G\}_\zeta & = & \delta^z_G
		\tilde\delta^\lambda F + \delta^z_{\tilde\delta^\lambda G} F  
		- \{\tilde
		\delta^\lambda F,G\} \nonumber \\
		&& -\d_{(i)}\d_k \left[\sum_{l=0}\sigma^{AB}(-\d)_{(j)}
				\left(\frac{\d \delta^l_t\Lambda^a}{\d
					z^B_{(j)}}\frac{\delta G}{\delta
				\overset{(l)}{\lambda^a}}\right)
			\frac{\delta F}{\delta z^A_{(i)k}}\right]\nonumber\\
		&& - \d_{(i)} \left[\sum_{l=0}\delta^z_G(\delta_t^l\Lambda^a)
			\frac{\delta F}{\delta
		\overset{(l)}{\lambda^a}_{(i)}}\right] - (F\leftrightarrow G).
	\end{eqnarray}
	Using the fact that the constraints and all their derivatives are zero
	on the boundary along with corollary \ref{theo:coroll}, this becomes
	equation \eqref{eq:applemma}.
\end{proof}

The cyclic identity we want to compute is between couples $(G_n,
\delta^\lambda_n)$. Let's define
\begin{equation}
	[[(G_1,\delta^\lambda_1),(G_2,\delta^\lambda_1)],(G_3,\delta^\lambda_1)]
	+ cyclic
	= (G_J,\delta^\lambda_J),
\end{equation}
where the bracket between couples is defined in equation \eqref{eq:GFTbracket}.
Using \eqref{eq:applemma} and the Jacobi identity proven in appendix 
\ref{sec:jacobiMB}, we obtain
\begin{eqnarray}
	G_J &=& \tilde\delta^\lambda_{[1,2]}G_3 - \tilde\delta^\lambda_3 \{G_2, G_1\}_\zeta
	- \tilde\delta^\lambda_1\tilde\delta^\lambda_2G_3 + \tilde\delta^\lambda_2 \tilde\delta^\lambda_1G_3 \nonumber \\
	&& \quad+
	\{G_2, \tilde\delta^\lambda_3G_1\}_\zeta +\{\tilde\delta^\lambda_3G_2, G_1 
	\}_\zeta + cyclic\nonumber\\
	&=&\sum_{l=0}\d_{(i)} \left[\left(\delta^l_t
	\Lambda^a_{[1,2]}+\tilde\delta_2(\delta_t^l\Lambda_1^a)
-\tilde\delta_1(\delta_t^l\Lambda_2^a)\right)
			\frac{\delta G_3}{\delta
		\overset{(l)}{\lambda^a}_{(i)}}\right]+ cyclic\nonumber\\
		&=&\sum_{l=0}\d_{(i)} \left[\left([\tilde\delta_2,\delta_t^l]\Lambda_1^a
	-[\tilde\delta_1,\delta_t^l]\Lambda_2^a\right)
			\frac{\delta G_3}{\delta
		\overset{(l)}{\lambda^a}_{(i)}}\right]+ cyclic
\end{eqnarray}

Playing with commutation relations, we obtain:
\begin{eqnarray}
	[\delta_t, \tilde\delta_2] \Lambda & = &\sigma^{AB}\d_{(i)}\left( (\delta_t
	\frac{\delta G_2}{\delta z^B} - \tilde \delta_2
\frac{\delta H}{\delta z^B})\frac{\delta \Lambda}{\delta z^A_{(i)}}\right)\\
&=& \d_{(i)} \left[\sigma^{AB}\frac{\delta}{\delta z^B}\left(\delta_tG_2 -
\tilde\delta_2H - \{G_2, H\}\right)\frac{\delta \Lambda}{\delta z^A_{(i)}}\right]\nonumber\\
&& -\d_{(i)}\left[ (-\d)_{(j)}\left(\frac{\d \Lambda^a_2}{\d
z^B_{(j)}}\phi_a\right)\sigma^{BA}\frac{\delta \Lambda}{\delta z^A_{(i)}}\right]\\
&=& -\d_{(i)}\left[ (-\d)_{(j)}\left(\frac{\d \Lambda^a_2}{\d
z^B_{(j)}}\phi_a\right)\sigma^{BA}\frac{\delta \Lambda}{\delta
z^A_{(i)}}\right],
\end{eqnarray}
where we used the fact that $(G_2, \delta^\lambda_2)$ generates an external
symmetry. This means that, on the constraint surface, we have:
\begin{equation}
	G_J \approx 0.
\end{equation}

Another useful result is that, if $(G_3, \delta^\lambda_3)$ is the hamiltonian
$(H, \delta_H^\lambda)$, we have
\begin{equation}
	G_J = 0.
\end{equation}

\newpage

\bibliography{bibli}

\providecommand{\href}[2]{#2}\begingroup\raggedright\begin{thebibliography}{10}

\bibitem{regge_role_1974}
T.~Regge and C.~Teitelboim, ``Role of surface integrals in the {Hamiltonian}
  formulation of general relativity,'' {\em Annals of Physics} {\bfseries 88}
  no.~1, (1974) 286--318.

\bibitem{brown_poisson_1986}
J.~D. Brown and M.~Henneaux, ``On the {Poisson} brackets of differentiable
  generators in classical field theory,'' {\em Journal of mathematical physics}
  {\bfseries 27} no.~2, (1986) 489--491.

\bibitem{arnowitt_republication_2008}
R.~Arnowitt, S.~Deser, and C.~W. Misner, ``Republication of: {The} dynamics of
  general relativity,'' {\em General Relativity and Gravitation} {\bfseries 40}
  no.~9, (2008) 1997--2027.

\bibitem{henneaux_asymptotically_2004}
M.~Henneaux, C.~Martínez, R.~Troncoso, and J.~Zanelli, ``Asymptotically
  anti–de {Sitter} spacetimes and scalar fields with a logarithmic branch,''
  {\em Physical Review D} {\bfseries 70} no.~4, (2004) 044034.

\bibitem{henneaux_asymptotic_2007}
M.~Henneaux, C.~Martínez, R.~Troncoso, and J.~Zanelli, ``Asymptotic behavior
  and {Hamiltonian} analysis of anti-de {Sitter} gravity coupled to scalar
  fields,'' {\em Annals of Physics} {\bfseries 322} no.~4, (2007) 824--848.

\bibitem{witten_anti-sitter_1998}
E.~Witten, ``Anti-de {Sitter} space and holography,'' {\em
  Adv.Theor.Math.Phys.} {\bfseries 2} (1998) 253--291.

\bibitem{de_azcarraga_lie_1998}
J.~A. de~Azcárraga and J.~M. Izquierdo, {\em Lie groups, {Lie} algebras,
  cohomology and some applications in physics}.
\newblock Cambridge University Press, 1998.

\bibitem{barnich_surface_2008}
G.~Barnich and G.~Compere, ``Surface charge algebra in gauge theories and
  thermodynamic integrability,'' {\em Journal of Mathematical Physics}
  {\bfseries 49} no.~4, (2008) 042901.

\bibitem{bering_putting_2000}
K.~Bering, ``Putting an edge to the {Poisson} bracket,'' {\em Journal of
  Mathematical Physics} {\bfseries 41} no.~11, (2000) 7468--7500.

\bibitem{henneaux_quantization_1992}
M.~Henneaux and C.~Teitelboim, {\em Quantization of gauge systems}.
\newblock Princeton university press, 1992.

\bibitem{troessaert_canonical_2013}
C.~Troessaert, ``Canonical {Structure} of {Field} {Theories} with {Boundaries}
  and {Applications} to {Gauge} {Theories},'' {\em arXiv preprint
  arXiv:1312.6427} (2013) .

\end{thebibliography}\endgroup

\end{document}